\documentclass[english]{article}
\usepackage[T1]{fontenc}
\usepackage[latin9]{inputenc}
\usepackage{geometry}
\usepackage{xcolor}
\usepackage{pdfcolmk}
\usepackage{array}
\usepackage{units}
\usepackage{multirow}
\usepackage{amsmath}
\usepackage{amsthm}
\usepackage{amssymb}
\usepackage{stmaryrd}
\usepackage{graphicx}
\usepackage[all]{xy}
\PassOptionsToPackage{normalem}{ulem}
\usepackage{ulem}

\makeatletter

\providecommand{\tabularnewline}{\\}
\providecolor{lyxadded}{rgb}{0,0,1}
\providecolor{lyxdeleted}{rgb}{1,0,0}

\newcommand{\lyxaddress}[1]{
\par {\raggedright #1
\vspace{1.4em}
\noindent\par}
}
\theoremstyle{plain}
\newtheorem{thm}{\protect\theoremname}[section]
  \theoremstyle{remark}
  \newtheorem{rem}[thm]{\protect\remarkname}
  \theoremstyle{definition}
  \newtheorem{defn}[thm]{\protect\definitionname}
  \theoremstyle{definition}
  \newtheorem{example}[thm]{\protect\examplename}
  \theoremstyle{plain}
  \newtheorem{cor}[thm]{\protect\corollaryname}
  \theoremstyle{plain}
  \newtheorem*{thm*}{\protect\theoremname}
  \theoremstyle{plain}
  \newtheorem*{cor*}{\protect\corollaryname}

\date{}
\usepackage[all]{xy}
\usepackage{placeins} 
\usepackage{lineno}
\usepackage{flushend}
\usepackage{graphicx}
\usepackage{changepage}
\usepackage{cite}
\usepackage[titletoc,title]{appendix}

\newcommand{\R}{\mathcal{R}}
\newcommand{\T}{\mathcal{T}}

\newcommand{\D}{\mathcal{D}}

\newcommand{\C}{\mathsf{C}}

\def\frontmatter@abstractheading{}

\renewcommand{\thesection}{\arabic{section}}

\makeatletter
\renewcommand{\p@subsection}{}
\renewcommand{\p@subsubsection}{}
\makeatother

\makeatother

\usepackage{babel}
  \providecommand{\corollaryname}{Corollary}
  \providecommand{\definitionname}{Definition}
  \providecommand{\examplename}{Example}
  \providecommand{\remarkname}{Remark}
  \providecommand{\theoremname}{Theorem}
\providecommand{\theoremname}{Theorem}

\begin{document}

\title{Contextuality in Canonical Systems of Random Variables}

\author{Ehtibar N. Dzhafarov,\textsuperscript{1} Víctor H. Cervantes,\textsuperscript{2 }and
Janne V. Kujala\textsuperscript{3 }}
\maketitle

\lyxaddress{\begin{center}
\textsuperscript{1}Purdue University, USA, ehtibar@purdue.edu \\
\textsuperscript{2}Purdue University, USA, cervantv@purdue.edu \\
\textsuperscript{3}University of Jyväskylä, Finland, jvk@iki.fi 
\par\end{center}}
\begin{abstract}
Random variables representing measurements, broadly understood to
include any responses to any inputs, form a system in which each of
them is uniquely identified by its content (that which it measures)
and its context (the conditions under which it is recorded). Two random
variables are jointly distributed if and only if they share a context.
In a canonical representation of a system, all random variables are
binary, and every content-sharing pair of random variables has a unique
maximal coupling (the joint distribution imposed on them so that they
coincide with maximal possible probability). The system is contextual
if these maximal couplings are incompatible with the joint distributions
of the context-sharing random variables. We propose to represent any
system of measurements in a canonical form and to consider the system
contextual if and only if its canonical representation is contextual.
As an illustration, we establish a criterion for contextuality of
the canonical system consisting of all dichotomizations of a single
pair of content-sharing categorical random variables.

\end{abstract}

\section{\label{sec: Introduction}Introduction}

We begin by recapitulating the basics of our theory of ``quantum-like''
contextuality, and then explain how this theory is developed in this
paper. The name of the theory is \emph{Contextuality-by-Default} (CbD),
and its recent accounts can be found in Ref. \cite{DzhafarovKujala(2016)Fortschritte,DzhafarovKujala(2017)LNCS,DzhafarovKujala(2016)Context-Content}. 
\begin{rem}
We use the following two notation conventions throughout the paper:
(1) due to its frequent occurrence we abbreviate the term \emph{random
variable} as \emph{rv} (\emph{rvs} in plural); and (2) we unconventionally
capitalize the words \emph{conteNt} and \emph{conteXt} to prevent
their confusion in reading.
\end{rem}
The matrix below represents the smallest possible version of what
we call a \emph{cyclic system} \cite{DzhafarovKujalaLarsson(2015),KujalaDzhafarov(2016)Proof,KujalaDzhafarovLar(2015),KujalaDzhafarov(2015)}:
\begin{center}
\begin{tabular}{|c|c|c}
\cline{1-2} 
$R_{1}^{1}$ & $R_{2}^{1}$ & $c=1$\tabularnewline
\cline{1-2} 
$R_{1}^{2}$ & $R_{2}^{2}$ & $c=2$\tabularnewline
\cline{1-2} 
\multicolumn{1}{c}{$q=1$} & \multicolumn{1}{c}{$q=2$} & $\boxed{\R}$\tabularnewline
\end{tabular}.
\par\end{center}

\noindent Each of the rvs $R_{q}^{c}$ represents measurements of
one of two properties, $q=1$ or $q=2$, under one of two conditions,
$c=1$ or $c=2$. The ``properties'' $q$ can also be called ``objects,''
``inputs,'' ``stimuli,'' etc., depending on the application, and
we refer to $q$ generically as the \emph{conteNt} of the measurement
$R_{q}^{c}$. The superscript $c$ in $R_{q}^{c}$ describes how and
under what circumstances $q$ is measured, including what other conteNts
are measured together with $q$. We refer to $c$ generically (and
traditionally) as the \emph{conteXt} of the measurement $R_{q}^{c}$.
The conteNt-conteXt pair $\left(q,c\right)$ provides a \emph{unique
identification} of $R_{q}^{c}$ within the system of measurements
$\R$. In addition, being an rv, $R_{q}^{c}$ is characterized by
its \emph{distribution}. In this paper, consideration is confined
to \emph{categorical rvs}, those with a finite number of values. The
term ``measurement'' is understood very broadly, to include any
response to any input or stimulus.

Let us begin with the simplest case of the system $\R$, when all
four rvs $R_{q}^{c}$ are \emph{binary}. In quantum physics, $R_{q}^{c}$
may describe a measurement of spin along one of two fixed axes, $q=1$
or $q=2$, in a spin-$\nicefrac{1}{2}$ particle. In psychology, $R_{q}^{c}$
may describe a response to one of two Yes-No questions, $q=1$ or
$q=2$. In both applications, in conteXt $c=1$ one measures first
$q=1$ and then $q=2$; in conteXt $c=2$ the measurements are made
in the opposite order. The rvs sharing a conteXt $c$ are recorded
in pairs, $\left(R_{1}^{c},R_{2}^{c}\right)$, which means that they
are \emph{jointly distributed} and can be viewed as a single (here,
four-valued) rv. No such joint distribution is defined for rvs in
different conteXts, such as $R_{2}^{1}$ and $R_{1}^{2}$. They are
\emph{stochastically unrelated} (to each other): one cannot ask about
the probability of an ``event'' $\left[R_{2}^{1}=x,R_{1}^{2}=y\right]$,
as no such ``event'' is defined. In particular, two conteNt-sharing
rvs, $R_{q}^{1}$ and $R_{q}^{2}$, are always stochastically unrelated,
hence they can never be considered one and the same rv, even if they
are identically distributed (see Ref. \cite{DzhafarovKujala(2016)Fortschritte}
for a detailed probabilistic analysis). 

In both applications mentioned, the distributions of $R_{q}^{1}$
and $R_{q}^{2}$ are de facto different. In the quantum-mechanical
example, the first spin measurement generally changes the state of
the particle \cite{Bacciagaluppi(2015)}. Assuming identical preparations
in both conteXts $c$, therefore, the state of the particle when a
$q$-spin is measured first will be different from that when it is
measured second. In the behavioral example, one's response to a question
asked second will generally be influenced by the question asked first
\cite{WangBusemeyer2013,Moore}. This creates obvious \emph{conteXt-dependence}
of the measurements, but this is not what we call \emph{contextuality}
in our theory. The original meaning of the term in quantum mechanics,
when translated into the language of probability theory (as in Refs.
\cite{DzhafarovKujala(2016)Context-Content,DzhafarovKujala(2016)Fortschritte,DzhafarovKujalaCervantes(2016)LNCS}
and, with caveats, \cite{Cabello(2013),KujalaDzhafarovLar(2015),KurzynskiRamanathanKaszlikowski(2012),Khr2009,Khr2005,Fine1982,SuppesZanotti1981}),
is that measurements of one and the same physical property $q$ have
to be represented by different rvs depending on what other properties
are being measured together with $q$ \textemdash{} even when the
laws of physics exclude all \emph{direct interactions} (energy/information
transfer) between the measurements. By extension, when such direct
interactions are present, as they are in our two applications of the
system $\R$, we speak of contextuality only if the dependence of
$R_{q}^{c}$ on $c$ is greater, in a well-defined sense, than just
the changes in its distribution. Contextuality is a non-causal aspect
of conteXt-dependence, revealed in the probabilistic relations between
different measurements rather than in their individual distributions.

This is how this understanding is implemented in CbD. We characterize
the conteXt-induced changes in the individual distributions, i.e.,
the difference between those of $R_{q}^{1}$ and $R_{q}^{2}$, by
\emph{maximally coupling} them. This means that we replace $R_{q}^{1}$
and $R_{q}^{2}$ with jointly distributed $T_{q}^{1}$ and $T_{q}^{2}$
that have the same respective individual distributions, and among
all such couplings we find one with the maximal value of $\Pr\left[T_{q}^{1}=T_{q}^{2}\right]$.
This maximal coupling $\left(T_{q}^{1},T_{q}^{2}\right)$ always exists
and is unique. The next step is to see if there exists an \emph{overall
coupling} $S$ of $\R$, a jointly distributed quadruple with elements
corresponding to those of $\R$,
\begin{center}
\begin{tabular}{|c|c|cc}
\cline{1-2} 
$S_{1}^{1}$ & $S_{2}^{1}$ & $c=1$ & \multirow{2}{*}{}\tabularnewline
\cline{1-2} 
$S_{1}^{2}$ & $S_{2}^{2}$ & $c=2$ & \tabularnewline
\cline{1-2} 
\multicolumn{1}{c}{$q=1$} & \multicolumn{1}{c}{$q=2$} & $\boxed{S}$ & \tabularnewline
\end{tabular},
\par\end{center}

\noindent such that its rows $\left(S_{1}^{c},S_{2}^{c}\right)$
are distributed as the rows of $\R$ and its columns $\left(S_{q}^{1},S_{q}^{2}\right)$
are distributed as the maximal couplings $\left(T_{q}^{1},T_{q}^{2}\right)$
of the columns of $\R$. If such a\emph{ maximally-connected} coupling
$S$ does not exist, one can say that the within-conteXt (row-wise)
relations prevent different measurements of the same conteNt (column-wise)
to be as close to each other as this is allowed by the direct influences
alone. Put differently, the relations of $R_{q}^{1}$ and $R_{q}^{2}$
with their same-conteXt counterparts force them, if imposed a joint
distribution upon, to coincide less frequently than if these relations
are ignored. The system then is deemed \emph{contextual}. Conversely,
if the coupling $S$ above exists, the within-conteXt relations do
not make the measurements of $R_{q}^{1}$ and $R_{q}^{2}$ any more
dissimilar than required by the direct influences: the system is \emph{noncontextual}.

The (non)existence of $S$ is determined by a simple linear programing
procedure \cite{DzhafarovKujalaLarsson(2015),DzhafarovKujala(2016)Context-Content}:
in our example, $S$ has $2^{4}$ possible values, and we find out
if they can be assigned nonnegative numbers (probability masses) that
sum to the given row-wise probabilities $\Pr\left[R_{1}^{c}=x,R_{2}^{c}=y\right]$
and the computed column-wise probabilities $\left[T_{q}^{1}=x,T_{q}^{2}=y\right]$.
There is also a simple criterion (inequality) for the existence of
a solution for this system of equations \cite{DzhafarovKujalaLarsson(2015),KujalaDzhafarov(2016)Proof,KujalaDzhafarovLar(2015)}.
Using it one can show, e.g., that in our quantum-mechanical application
the system $\R$ is always noncontextual, and so it is in the behavioral
application if one adopts the model proposed in Ref. \cite{WangBusemeyer2013}
(see Ref. \cite{DzhafarovZhangKujala(2015)Isthere} for details).
Mathematically, however, the system $\R$ can be contextual, and if
it is, CbD provides a simple way of computing the \emph{degree} of
its contextuality \cite{DzhafarovKujala(2016)Context-Content}: one
replaces the probability masses in the above linear programing task
with \emph{quasi-probabilities}, allowed to be negative, and finds
among the solutions the minimum sum of their absolute values (see
Section \ref{subsec: Degree-of-contextuality}). 

Although most of these principles and procedures of CbD have been
formulated for arbitrary systems of measurements \cite{DzhafarovKujalaCervantes(2016)LNCS,DzhafarovKujala(2016)Context-Content},
they only work without complications with systems that satisfy the
following two constraints: (A) they contain only binary rvs, and (B)
there are no more than two rvs sharing a conteNt (i.e., occupying
the same column). What we propose in this paper is to always present
a system of measurements in a \emph{canonical form}, which is in essence
one with the properties A and B. The cyclic systems form a subclass
of canonical systems, rich enough to cover most experimental paradigms
of traditional interest in quantum-mechanical and behavioral contextuality
studies \cite{DzhafarovKujala(2016)Context-Content,DzhafarovKujalaCervantes(2016)LNCS,DzhafarovKujalaCervantesZhangJones(2016),DzhafarovKujalaLarsson(2015),DzhafarovZhangKujala(2015)Isthere,KujalaDzhafarovLar(2015)},
but far from satisfactory generality. 

What are the complications one faces if a system does not satisfy
the properties A and B? Consider the system below, with all its rvs
binary but with three rather than two of them in each column:
\begin{center}
\begin{tabular}{|c|c|cc}
\cline{1-2} 
$R_{1}^{1}$$\ensuremath{}$ & $R_{2}^{1}$$\ensuremath{}$ & $c=1$ & \multirow{3}{*}{}\tabularnewline
\cline{1-2} 
$R_{1}^{2}$$\ensuremath{}$ & $R_{2}^{2}$$\ensuremath{}$ & $c=2$ & \tabularnewline
\cline{1-2} 
$R_{1}^{3}$$\ensuremath{}$ & $R_{2}^{3}$$\ensuremath{}$ & $c=3$ & \tabularnewline
\cline{1-2} 
\multicolumn{1}{c}{$q=1$} & \multicolumn{1}{c}{$q=2$} & $\boxed{\R'}$ & \tabularnewline
\end{tabular}
\par\end{center}

\noindent How does CbD apply here? In the earlier version of theory
(summarized in Refs. \cite{DzhafarovKujala(2016)Context-Content,DzhafarovKujalaCervantes(2016)LNCS})
we computed the couplings $\left(T_{q}^{1},T_{q}^{2},T_{q}^{3}\right)$
of each column that maximize $\Pr\left[T_{q}^{1}=T_{q}^{2}=T_{q}^{3}\right]$.
One problem with this approach is that the maximal coupling $\left(T_{q}^{1},T_{q}^{2},T_{q}^{3}\right)$,
while it always exists, is not defined uniquely. What should be the
contextuality analysis of $\R'$ if the within-conteXt (row-wise)
distributions are compatible with some but not all combinations of
the maximal couplings for the two columns? Shall one then speak of
a partial (non)contextuality? Originally we proposed to consider a
system noncontextual if it is compatible with at least one of these
pairs of maximal couplings, but in addition to being arbitrary, this
leads to another complication: it may then very well happen that the
system $\R'$ is noncontextual but one of its subsystems, e.g. $\R$,
is contextual. This is contrary to one's intuition of noncontextuality. 

In the most recent publications therefore \cite{DzhafarovKujala(2017)LNCS,DzhafarovKujala(2016)Fortschritte}
we modified our approach into ``CbD 2.0,'' by positing that a coupling
for conteNt-sharing measurements should be computed so that it maximizes
the probability of coincidence for every \emph{pair} (equivalently,
every subset) of them. In our case, this means maximization of $\Pr\left[T_{q}^{1}=T_{q}^{2}\right]$,
$\Pr\left[T_{q}^{2}=T_{q}^{3}\right]$, and $\Pr\left[T_{q}^{1}=T_{q}^{3}\right]$
(it is in fact sufficient to maximize only certain pairs rather than
all of them, but this is not critical here). Such a coupling $\left(T_{q}^{1},T_{q}^{2},T_{q}^{3}\right)$
is called \emph{multimaximal}. With only binary rvs involved, a multimaximal
coupling always exists and is unique; and a subsystem of a noncontextual
system then is always noncontextual. 

Returning to system $\R$, consider now the situation when the measurements
involved are not dichotomous. For example, let the two successive
spin measurements along axes $q=1$ and $q=2$ be made on a hypothetical
spin-$2$ particle, with the measurement outcomes denoted $\left\{ -2,-1,0,1,2\right\} $.
In the behavioral application, let the questions asked allow 5 answers
each, labeled in the same way. A maximal coupling in this situation
exists for each column of $\R$, but not uniquely. This takes us back
to the problem of what one should do if the row-wise distributions
are compatible with some but not all pairs of these maximal couplings.
Another problem is even harder. If the system is deemed noncontextual,
one may consider it desirable that it remain noncontextual after some
of the measurement outcomes are ``lumped together.'' Thus, one may
wish to consider $\left\{ -2,-1,0,1,2\right\} $ in terms of ``negative-zero-positive,''
lumping together $-2$ with $-1$ and $2$ with $1$. Or one may wish
to look at the outcomes in terms of ``zero-nonzero.'' As it turns
out, a noncontextual system may become contextual after such coarsening
of some of its measurements. 

Both these problems can be resolved if we agree that \emph{every measurement
included in the system, empirically recorded or computed from those
empirically recorded, should be represented by a set of binary rvs}.
Let us denote by $D_{qW}^{c}$ the Bernoulli rv that equals 1 if the
value of $R_{q}^{c}$ is within the subset $W$ of its possible values.
We call $D_{qW}^{c}$ a \emph{split} (of the original rv). We posit
that a measurement with $k$ distinct values should always be represented
by $k$ ``detectors'' of these values, i.e. the splits with one-element
subsets $W$. Thus, in our system $\R$, each measurement $R_{q}^{c}$
should be replaced with the jointly distributed splits 
\[
\left(D_{q\left\{ -2\right\} }^{c},D_{q\left\{ -1\right\} }^{c},D_{q\left\{ 0\right\} }^{c},D_{q\left\{ 1\right\} }^{c},D_{q\left\{ 2\right\} }^{c}\right).
\]
If one is also interested in the coarsening of $R_{q'}^{c}$ into
values ``negative-zero-positive,'' then the list should be expanded
into
\[
\left(D_{q\left\{ -2\right\} }^{c},D_{q\left\{ -1\right\} }^{c},D_{q\left\{ 0\right\} }^{c},D_{q\left\{ 1\right\} }^{c},D_{q\left\{ 2\right\} }^{c},D_{q\left\{ -2,-1\right\} }^{c},D_{q\left\{ 1,2\right\} }^{c}\right).
\]
If one wishes to include \emph{all} possible coarsenings of the original
rvs in $\R$, then the set of binary rvs should consist of \emph{all}
possible splits. Since every dichotomization creating a split should
be applied to all rvs sharing a conteNt, one ends up replacing the
system $\R$ with
\begin{center}
\begin{tabular}{|c|c|c|c|c|c|c|c|c}
\cline{1-8} 
$D_{1\left\{ -2\right\} }^{1}$ & $\cdots$ & $D_{1\left\{ 2\right\} }^{1}$ & $D_{1\left\{ -2,-1\right\} }^{1}$ & $\cdots$ & $D_{1\left\{ 1,2\right\} }^{1}$ & $\cdots$ & $D_{2\left\{ 1,2\right\} }^{1}$ & $c=1$\tabularnewline
\cline{1-8} 
$D_{1\left\{ -2\right\} }^{2}$ & $\cdots$ & $D_{1\left\{ 2\right\} }^{2}$ & $D_{1\left\{ -2,-1\right\} }^{2}$ & $\cdots$ & $D_{1\left\{ 1,2\right\} }^{2}$ & $\cdots$ & $D_{2\left\{ 1,2\right\} }^{2}$ & $c=2$\tabularnewline
\cline{1-8} 
\multicolumn{1}{c}{$q=1\left\{ -2\right\} $} & \multicolumn{1}{c}{$\cdots$} & \multicolumn{1}{c}{$q=1\left\{ 2\right\} $} & \multicolumn{1}{c}{$q=1\left\{ -2,-1\right\} $} & \multicolumn{1}{c}{$\cdots$} & \multicolumn{1}{c}{$q=1\left\{ 1,2\right\} $} & \multicolumn{1}{c}{$\cdots$} & \multicolumn{1}{c}{$q=2\left\{ 1,2\right\} $} & $\boxed{\D}$\tabularnewline
\end{tabular}
\par\end{center}

\noindent There are $(2^{5}-2)/2=15$ distinct dichotomizations of
the set $\left\{ -2,-1,0,1,2\right\} $, and the 15 subsets $W$ in
$D_{qW}^{c}$ should be chosen to avoid duplication, such as in $D_{q\left\{ 0,1\right\} }^{c}$
and $D_{q\left\{ -2,-1,2\right\} }^{c}$. Once duplication is prevented,
however, all splits of all rvs one is interested in should be included.
It is irrelevant that some of them can be presented as functions of
the others. In fact, any split of our $R_{q}^{c}$ can be presented
as a function of just three splits, chosen, e.g., as
\[
D_{q'}^{c}=D_{q\left\{ -1,1\right\} }^{c},D_{q''}^{c}=D_{q\left\{ 0,1\right\} }^{c},D_{q'''}^{c}=D_{q\left\{ 2\right\} }^{c}.
\]
It is easy to show, however, that in the subsystem
\begin{center}
\begin{tabular}{|c|c|c||c|c}
\cline{1-4} 
$D_{1'}^{1}$$ $ & $D_{1''}^{1}$$ $ & $D_{1'''}^{1}$$ $ & $f\left(D_{1'}^{1},D_{1''}^{1},D_{1'''}^{1}\right)$$ $ & $c=1$\tabularnewline
\cline{1-4} 
$D_{1'}^{2}$$ $ & $D_{1''}^{2}$$ $ & $D_{1'''}^{2}$$ $ & $f\left(D_{1'}^{2},D_{1''}^{2},D_{1'''}^{2}\right)$$ $ & $c=2$\tabularnewline
\cline{1-4} 
\multicolumn{1}{c}{$q=1'$} & \multicolumn{1}{c}{$q=1''$} & \multicolumn{1}{c}{$q=1'''$} & \multicolumn{1}{c}{$q^{*}$} & $\boxed{\D'}$\tabularnewline
\end{tabular}
\par\end{center}

\noindent of the system $\D$, the $f$-transformation of the maximal
couplings of the first three columns, since these couplings are not
jointly distributed, would not determine the coupling of the fourth
column, let alone ensure that this coupling is maximal.

There is no general prescription as to which rvs should or should
not be included in the system representing an empirical set of measurements:
what one includes (e.g., what coarsenings of the rvs already in play
one considers) reflects what aspects of the empirical situation one
is interested in. Once a set of rvs is chosen, however, we uniquely
form their splits and place them in a canonical system. 

The remainder of the paper is organized as follows. In Section \ref{sec: Formal-Theory},
we present the abstract version of CbD applicable to all possible
systems of categorical (and not only categorical) rvs. In Section
\ref{sec: Split-and-Canonical}, we formalize the idea of representing
any system of rvs by their splits and applying contextuality analysis
to these representations only. In Section \ref{sec: A-systematic-study},
we investigate the representation of all coarsenings of a single pair
of conteNt-sharing rvs by all possible splits. In the concluding section
we explain why one might wish to consider only some rather than all
possible splits. 
\begin{rem}
\label{rem: The-proofs-of}The proofs of the formal propositions in
the paper, unless obvious or referenced as presented elsewhere, are
given in the supplementary file S, together with additional theorems
and examples. 
\end{rem}

\section{\label{sec: Formal-Theory}Formal Theory of Contextuality}

\subsection{\label{subsec: Basic-notions}Basic notions}

The definition of a system of rvs requires two nonempty finite sets,
a set of \emph{conteNts} $Q$ and a set of \emph{conteXts} $C$. There
is a relation 
\begin{equation}
\Yleft\subseteq Q\times C,
\end{equation}
such that the projections of $\Yleft$ into $Q$ and $C$ equal $Q$
and $C$, respectively (this means that for every $q\in Q$ there
is a $c\in C$, and vice versa, such that $q\Yleft c$). We read both
$q\Yleft c$ and $c\Yright q$ as ``$q$ \emph{is measured in} $c$.'' 

A categorical rv is one with a finite set of values and its power
set as the codomain sigma-algebra. A system of (categorical) rvs is
a double-indexed set (we use calligraphic letters for sets of random
variables)
\begin{equation}
\R=\left\{ R_{q}^{c}:q\in Q,c\in C,q\Yleft c\right\} ,\label{eq: cc-system}
\end{equation}
such that (i) any $R_{q}^{c}$ and $R_{q}^{c'}$ have the same set
of possible values; (ii) $R_{q}^{c}$ and $R_{q'}^{c'}$ are jointly
distributed if $c=c'$; and (iii) if $c\not=c'$, $R_{q}^{c}$ and
$R_{q'}^{c'}$ are \emph{stochastically unrelated} (possess no joint
distribution). For any $c\in C$ the subset 
\begin{equation}
\R^{c}=\left\{ R_{q}^{c}:q\in Q,q\Yleft c\right\} =R^{c}
\end{equation}
of $\R$ is called a \emph{bunch} (of rvs) corresponding to $c$.
Since the elements of a bunch are jointly distributed, the bunch is
a (categorical) rv in its own right, so it can be also written as
$R^{c}$. Note that we do not distinguish the representations of $\R$
as (\ref{eq: cc-system}) and as 
\begin{equation}
\R=\left\{ R^{c}:c\in C\right\} .
\end{equation}
(See Refs. \cite{DzhafarovKujala(2016)Context-Content,DzhafarovKujala(2016)Fortschritte}
for a detailed probabilisitic analysis.) 

For any $q\in Q$, the subset 
\begin{equation}
\R_{q}=\left\{ R_{q}^{c}:c\in C,q\Yleft c\right\} 
\end{equation}
of $\R$ is called a \emph{connection} (between the bunches of rvs)
corresponding to $q$. Any two elements of a connection are stochastically
unrelated, so it is not an rv. 

\subsection{\label{subsec: General-definition-of-contextuality}General definition
of (non)contextuality}

A (probabilistic) \emph{coupling }$Y$ of a set of rvs $\left\{ X_{1},\ldots,X_{n}\right\} $
is a set of jointly distributed $\left\{ Y_{1},\ldots,Y_{n}\right\} $
such that $Y_{i}\sim X_{i}$ for $i=1,\ldots,n$. The tilde $\sim$
stands for ``has the same distribution as.'' 

An (overall) coupling $S$ of a system $\R$ in (\ref{eq: cc-system})
is a coupling of its bunches. That is, it is an rv
\begin{equation}
S=\left\{ S^{c}:c\in C\right\} 
\end{equation}
(with jointly distributed components) such that $S^{c}\sim R^{c}$,
for any $c\in C$. This implies that 
\begin{equation}
S^{c}=\left\{ S_{q}^{c}:q\in Q,q\Yleft c\right\} 
\end{equation}
is a set of jointly distributed rvs in a one-to-one correspondence
with the identically labeled elements of $\R$.

For a given $q\in Q$, a coupling $T_{q}$ of a connection $\R_{q}$
is an rv 
\begin{equation}
T_{q}=\left\{ T_{q}^{c}:c\in C,q\Yleft c\right\} 
\end{equation}
such that $T_{q}^{c}\sim R_{q}^{c}.$ In particular, if $S$ is a
coupling of $\R$, then 
\begin{equation}
S_{q}=\left\{ S_{q}^{c}:c\in C,q\Yleft c\right\} 
\end{equation}
is a coupling of $\R_{q}$, for any $q\in Q$. 
\begin{defn}
Given a set $\T=\left\{ T^{c}:c\in C\right\} $ of couplings for all
connections in a system $\R$, the system is said to be \emph{noncontextual
with respect to} $\T$ if $\R$ has a coupling $S$ with $S_{q}\sim T_{q}$
for any $q\in Q$. Otherwise $\R$ is said to be \emph{contextual
with respect to} $\T$. 
\end{defn}
Put differently, $\R$ is noncontextual with respect to $\T$ if and
only if there is a jointly distributed set 
\begin{equation}
S=\left\{ S_{q}^{c}:q\in Q,c\in C,q\Yleft c\right\} ,
\end{equation}
such that, for every $c\in C$, $S^{c}\sim R^{c}$, and for every
$q\in Q$, $S_{q}\sim T_{q}$. A coupling $S$ with this property
is called $\T$-\emph{connected}. 

If the couplings $T_{q}$ are characterized by some property $\C$
such that one and only one coupling $T_{q}$ satisfies this property
for any given connection $\R_{q}$, then the definition can be rephrased
as follows: 
\begin{defn}
$\R$ is said to be \emph{noncontextual with respect to} \emph{property}
$\C$ if it has a $\C$-\emph{connected} coupling $S$, defined as
one with $S_{q}$ satisfying $\C$ for any $q\in Q$. Otherwise $\R$
is said to be \emph{contextual with respect to} $\C$.
\end{defn}
\begin{rem}
\label{rem: multimaximally connected}In Section \ref{subsec: Multimaximality-for-splits}
we will use the property of \emph{(multi)maximality} to play the role
of $\C$, and the couplings in question then are referred to as \emph{(multi)maximally-connected}.
\end{rem}

\subsection{\label{subsec: Degree-of-contextuality}Degree of contextuality}

A \emph{quasi-distribution} on a finite set $V$ is a function $V\rightarrow\mathbb{R}$
(real numbers) such that the numbers assigned to the elements of $V$
sum to 1. We will refer to these numbers as \emph{quasi-probability
masses}. A \emph{quasi-rv} $X$ is defined analogously to an rv but
with a quasi-distribution instead of a distribution.

A \emph{quasicoupling} $X$ of $\R$ is defined as a quasi-rv 
\begin{equation}
X=\left\{ X_{q}^{c}:q\in Q,c\in C,q\Yleft c\right\} ,
\end{equation}
such that $X^{c}\sim R^{c}$ for every $c\in C$. We have the following
results.
\begin{thm}[\cite{DzhafarovKujala(2016)Context-Content} Theorem 6.1]
\label{thm: quasi-couplings}For any system $\R$ and any set $\T$
of couplings for the connections of $\R$, there is a quasi-coupling
$X$ of $\R$ such that $X_{q}=\left\{ X_{q}^{c}:c\in C,q\Yleft c\right\} \sim T_{q}$
for any $q\in Q$. 
\end{thm}
The \emph{total variation} of $X$ is denoted by $\left\Vert X\right\Vert $
and defined as the sum of the absolute values of the quasi-probability
masses assigned to all values of $X$. 
\begin{thm}[\cite{DzhafarovKujala(2016)Context-Content} Section 6.3]
The total variation $\left\Vert X\right\Vert $ reaches its minimum
in the class of all quasi-couplings $X$ satisfying the conditions
of Theorem \ref{thm: quasi-couplings}.
\end{thm}
If $\min\left\Vert X\right\Vert $ is 1, then all quasi-probability
masses are nonnegative, and the system $\R$ is noncontextual with
respect to $\T$. If $\min\left\Vert X\right\Vert >1$, then the system
is contextual with respect to $\T$, and $\min\left\Vert X\right\Vert -1$
can be taken as a (universally applicable) measure\emph{ }of the \emph{degree
of contextuality}.

\section{\label{sec: Split-and-Canonical}Splits and Canonical Representations}

\subsection{\label{subsec: Expansions}Expansions of the original system}

One is often interested not only in a system of empirically measured
rvs $\R$ but also in some transformations thereof. Each such a transformation
$F_{q_{1},\dots,q_{k}}$ is labeled by a set of conteNts, $q_{1},\dots,q_{k}$,
and it takes as its arguments the rvs $R_{q_{1}}^{c},\ldots,R_{q_{k}}^{c}$
in each conteXt $c$ such that $c\Yright q_{1},\ldots,q_{k}$. The
outcome, 
\begin{equation}
R_{q^{*}}^{c}=F_{q_{1},\dots,q_{k}}\left(R_{q_{1}}^{c},\ldots,R_{q_{k}}^{c}\right),
\end{equation}
is an rv interpreted as measuring a new conteNt $q^{*}$ in the conteXt
$c$. One is free to choose any such transformations and form the
corresponding new conteNts, as there can be no rules mandating what
one should be interested in measuring. 

Using various transformations to add new conteNts and new rvs to the
original system \emph{expands} it into a larger system. Two types
of expansions that are of particular interest are \emph{expansion-through-joining}
and \emph{expansion-through-coarsening}. Joining is defined as 
\begin{equation}
R_{q_{1}}^{c},\ldots,R_{q_{k}}^{c}\longmapsto\left(R_{q_{1}}^{c},\ldots,R_{q_{k}}^{c}\right)=R_{q'}^{c},
\end{equation}
whereas coarsening is transformation 
\begin{equation}
R_{q}^{c}\longmapsto F_{q}\left(R_{q}^{c}\right)=R_{q''}^{c}.
\end{equation}
In fact any other transformation $F_{q_{1},\dots,q_{k}}\left(R_{q_{1}}^{c},\ldots,R_{q_{k}}^{c}\right)$
can be presented as joining followed by coarsening. 
\begin{example}[Joining]
\label{exa: Joining}Consider the system

\begin{center}

\begin{tabular}{|c|c|c|c}
\cline{1-3} 
$R_{1}^{1}$$\ensuremath{}$ & $R_{2}^{1}$$\ensuremath{}$ & $\cdot$$\ensuremath{}$ & $c=1$\tabularnewline
\cline{1-3} 
$R_{1}^{2}$$\ensuremath{}$ & $R_{2}^{2}$$\ensuremath{}$ & $\cdot$$\ensuremath{}$ & $c=2$\tabularnewline
\cline{1-3} 
$R_{1}^{3}$$\ensuremath{}$ & $\cdot$$\ensuremath{}$ & $R_{3}^{3}$$\ensuremath{}$ & $c=3$\tabularnewline
\cline{1-3} 
$\cdot$$\ensuremath{}$ & $R_{2}^{4}$$\ensuremath{}$ & $R_{3}^{4}$$\ensuremath{}$ & $c=4$\tabularnewline
\cline{1-3} 
\multicolumn{1}{c}{$q=1$} & \multicolumn{1}{c}{$q=2$} & \multicolumn{1}{c}{$q=3$} & $\boxed{\R}$\tabularnewline
\end{tabular}.

\end{center}

\noindent It contains the jointly distributed $R_{1}^{1},R_{2}^{1}$
and also the jointly distributed $R_{1}^{2},R_{2}^{2}$, but in determining
the maximal couplings of $R_{1}^{1},R_{1}^{2}$ and of $R_{2}^{1},R_{2}^{2}$
in the first and second columns these row-wise joints are not utilized.
In some applications this would be unacceptable (e.g., in the theory
of selective influences \cite{TNHMP,DK2012JMP} and in the approach
advocated by Abramsky and colleagues \cite{AbramskyBarbosaKishidaLalMansfield(2015),AbramskyBrandenburger(2011)}
this is never acceptable), and then the following expansion has to
be used:

\begin{center}

\begin{tabular}{|c|c|c|c|r}
\cline{1-4} 
$R_{1}^{1}$$ $ & $R_{2}^{1}$$ $ & $\cdot$$ $ & $\left(R_{1}^{1},R_{2}^{1}\right)$ & $c=1$\tabularnewline
\cline{1-4} 
$R_{1}^{2}$$ $ & $R_{2}^{2}$$ $ & $\cdot$$ $ & $\left(R_{1}^{2},R_{2}^{2}\right)$ & 2\tabularnewline
\cline{1-4} 
$R_{1}^{3}$$ $ & $\cdot$$ $ & $R_{3}^{3}$$ $ & $\cdot$$ $ & 3\tabularnewline
\cline{1-4} 
$\cdot$$ $ & $R_{2}^{4}$$ $ & $R_{3}^{4}$$ $ & $\cdot$$ $ & 4\tabularnewline
\cline{1-4} 
\multicolumn{1}{c}{$q=1$} & \multicolumn{1}{c}{2} & \multicolumn{1}{c}{3} & \multicolumn{1}{c}{12} & $\boxed{\R^{*}}$\tabularnewline
\end{tabular}.$\square$

\end{center}
\end{example}
\begin{example}[Coarsening]
\label{exa: coarsening} If $V$ is a set of possible values of $R_{q}^{c}$,
then $U=F_{q}\left(V\right)$ is the set of possible values of the
rv $R_{q^{*}}^{c}=F_{q}\left(R_{q}^{c}\right).$ This rv is a coarsening
of $R_{q}^{c}$. Note that any rv is its own coarsening. Since the
way one labels the values of $U$ is usually irrelevant, each such
function $F_{q}$ can be presented as a partition of $V$. Consider,
e.g., the ``mini''-system

\begin{center}

\begin{tabular}{|c|c}
\cline{1-1} 
$R_{q}^{1}$$\ensuremath{}$ & $c=1$\tabularnewline
\cline{1-1} 
$R_{q}^{2}$$\ensuremath{}$ & $c=2$\tabularnewline
\cline{1-1} 
\multicolumn{1}{c}{$q$} & $\boxed{\R}$\tabularnewline
\end{tabular},

\end{center}

\noindent and let the two rvs take values on $\left\{ 1,2,3,4,5\right\} $.
If these values are considered ordered, $1<\ldots<5$, one may be
interested in all possible partitions of $\left\{ 1,2,3,4,5\right\} $
into subsets of consecutive numbers, such as $\left\{ 12\:|\:34\:|\:5\right\} $,
$\left\{ 1\:|\:2345\right\} $, etc. There are 15 such partitions
(counting $\left\{ 1\:|\:2\:|\:3\:|\:4\:|\:5\right\} $ that defines
the original rvs $R_{q}^{c}$, but excluding the trivial partition
$\left\{ 12345\right\} $). If the values $1,2,3,4,5$ are treated
as unordered labels, one might consider all possible nontrivial partitions,
such as $\left\{ \left\{ 14\right\} ,\left\{ 25\right\} ,\left\{ 3\right\} \right\} $,
$\left\{ \left\{ 145\right\} ,\left\{ 23\right\} \right\} $, etc.
There are 51 such partitions. In either of these two coarsening schemes
the partitions can be ordered in some way, and the respective expanded
systems then become

\begin{center}

\begin{tabular}{|c|c|c|c|c}
\cline{1-4} 
$R_{q}^{1}$ & $R_{q1'}^{1}$ & $\cdots$ & $R_{q14'}^{1}$ & $c=1$\tabularnewline
\cline{1-4} 
$R_{q}^{2}$ & $R_{q1'}^{2}$ & $\cdots$ & $R_{q14'}^{2}$$\ensuremath{}$ & $c=2$\tabularnewline
\cline{1-4} 
\multicolumn{1}{c}{$q$} & \multicolumn{1}{c}{$q1'$} & \multicolumn{1}{c}{$\cdots$} & \multicolumn{1}{c}{$q14'$} & $\boxed{\R'}$\tabularnewline
\end{tabular} and$\quad$%
\begin{tabular}{|c|c|c|c|c}
\cline{1-4} 
$R_{q}^{1}$ & $R_{q1''}^{1}$ & $\cdots$ & $R_{q50''}^{1}$ & $c=1$\tabularnewline
\cline{1-4} 
$R_{q}^{2}$ & $R_{q1''}^{2}$ & $\cdots$ & $R_{q50''}^{2}$ & $c=2$\tabularnewline
\cline{1-4} 
\multicolumn{1}{c}{$q$} & \multicolumn{1}{c}{$q1''$} & \multicolumn{1}{c}{$\cdots$} & \multicolumn{1}{c}{$q50''$} & $\boxed{\R''}$\tabularnewline
\end{tabular}

\end{center}
\end{example}
\begin{rem}
\label{rem: support small}Although the number of the states (combinations
of the values of the elements) of the bunch $R^{c}$ in $\R'$ and
especially in $\R^{''}$ is very large, the support of each bunch
(the set of the states with nonzero probabilities) has the same size
as that of the initial random variable $R_{q}^{c}$ in $\R$ (i.e.,
in our example, it cannot exceed 5). This follows from the facts that
each event $R_{q}^{c}=x$ uniquely defines the state of $R^{c}$ in
$\R'$ and in $\R^{''}$, and that $\sum_{x}\Pr\left[R_{q}^{c}=x\right]=1$.
\hfill$\square$
\end{rem}

\subsection{\label{subsec: Dichotomizations-and-splits}Dichotomizations and
canonical/split representations}
\begin{defn}
A \emph{dichotomization} of a set $V$ is a function $f:V\rightarrow\left\{ 0,1\right\} $.
Applying such an $f$ to an rv $R$ with the set of possible values
$V$, we get a binary rv $f\left(R\right)$. We call this $f\left(R\right)$
a \emph{split} of the original $R$. 
\end{defn}
If $R_{q}^{c}$ is an element of a system $\R$, let us agree to identify
$f\left(R_{q}^{c}\right)$ as $D_{qW}^{c}$, where $W=f^{-1}\left(1\right)$,
with the understanding that $D_{qW}^{c}$ and $D_{q\left(V-W\right)}^{c}$
are indistinguishable. To make the choice definitive, we always choose
$W$ as the smaller of $W$ and $V-W$; in the case they have the
same number of elements, we order the elements of $V$, say $1<2<\ldots<k$,
and then choose $W$ as lexicographically preceding $V-W$.

With $V=\left\{ 1,2,\ldots,k\right\} $, the jointly distributed set
of splits 
\begin{equation}
\left\{ D_{q\left\{ 1\right\} }^{c},D_{q\left\{ 2\right\} }^{c},\ldots,D_{q\left\{ k\right\} }^{c}\right\} 
\end{equation}
 is called the \emph{split representation} of $R_{q}^{c}$. If $k=2$,
then $R_{q}^{c}$ is its own split representation, because $D_{q\left\{ 1\right\} }^{c}$
and $D_{q\left\{ 2\right\} }^{c}$ are indistinguishable. 
\begin{defn}
The system $\D$ obtained from a system $\R$ by replacing each of
its elements by its split representations is called the \emph{canonical
}(or \emph{split})\emph{ representation} of $\R$.
\end{defn}
\begin{example}[continuing Example \ref{exa: Joining}]
\label{exa: Joining splits} Let all rvs in $\R$ be binary, $0/1$,
whence $\left(R_{1}^{1},R_{2}^{1}\right)$ and $\left(R_{1}^{2},R_{2}^{2}\right)$
in $\R^{*}$ have 4 values each: $00,01,10,11$. Replacing them with
the split representations and observing that the first three columns
do not change, we get the following canonical representation of $\R^{*}$:

\begin{center}

\begin{tabular}{|c|c|c|c|c|c|c|r}
\cline{1-7} 
$D_{1}^{1}=R_{1}^{1}$$ $ & $D_{2}^{1}=R_{2}^{1}$$ $ & $\cdot$$ $ & $D_{12\left\{ 00\right\} }^{1}$ & $D_{12\left\{ 01\right\} }^{1}$ & $D_{12\left\{ 10\right\} }^{1}$ & $D_{12\left\{ 11\right\} }^{1}$ & $c=1$\tabularnewline
\cline{1-7} 
$D_{1}^{2}=R_{1}^{2}$$ $ & $D_{2}^{2}=R_{2}^{2}$$ $ & $\cdot$$ $ & $D_{12\left\{ 00\right\} }^{2}$ & $D_{12\left\{ 01\right\} }^{2}$ & $D_{12\left\{ 10\right\} }^{2}$ & $D_{12\left\{ 11\right\} }^{2}$ & $2$\tabularnewline
\cline{1-7} 
$D_{1}^{3}=R_{1}^{3}$$ $ & $\cdot$$ $ & $D_{3}^{3}=R_{3}^{3}$$ $ & $\cdot$$ $ & $\cdot$$ $ & $\cdot$$ $ & $\cdot$$ $ & $3$\tabularnewline
\cline{1-7} 
$ $ & $D_{2}^{4}=R_{2}^{4}$$ $ & $D_{3}^{4}=R_{3}^{4}$$ $ & $\cdot$$ $ & $\cdot$$ $ & $\cdot$$ $ & $\cdot$$ $ & $4$\tabularnewline
\cline{1-7} 
\multicolumn{1}{c}{$q=1$} & \multicolumn{1}{c}{$2$} & \multicolumn{1}{c}{$3$} & \multicolumn{1}{c}{$12\left\{ 00\right\} $} & \multicolumn{1}{c}{$12\left\{ 01\right\} $} & \multicolumn{1}{c}{$12\left\{ 10\right\} $} & \multicolumn{1}{c}{$12\left\{ 11\right\} $} & $\boxed{\D^{*}}$\tabularnewline
\end{tabular}.$\square$

\end{center}
\end{example}
\begin{example}[continuing Example \ref{exa: coarsening}]
\label{exa: coarsening splits} For the system $\R'$, it is clear
that the split representations of the 15 coarsenings of $R_{q}^{c}$
multiply overlap: e.g., $D_{q\left\{ 3\right\} }^{1}$ belongs to
the split representations of $R_{q}^{1}$ and of the coarsenings defined
by the partitions $\left\{ 12\:|\:3\:|\:45\right\} $, $\left\{ 1\:|\:2\:|\:3\:|\:45\right\} $,
and $\left\{ 12\:|\:3\:|\:4\:|\:5\right\} $. Following our rules,
$W$ in the splits $D_{qW}^{c}$ comprising the split representation
of $\R'$ are (when written as strings) $1,2,3,4,5,12,23,34,45,$
and $15$ (note that, e.g., the split of the coarsening \{1|23|4|5\}
with $W=\{1,23\}$ should be denoted $D_{q\{1,23\}}^{1}$ according
to our definitions, but this is the same random variable as $D_{q\{45\}}^{1}$
which we have included in the list). For the system $\R''$ the canonical
representation, obviously, consists of all possible splits of $R_{q}^{c}$.
It will be the target of the analysis presented in Section \ref{sec: A-systematic-study}.\hfill$\square$
\end{example}

\subsection{\label{subsec: Multimaximality-for-splits}Multimaximality for canonical
representations}

If each connection in a canonical representation $\D$ contains just
two rvs, one can compute unique maximal couplings for all of these
connections. The determination of whether $\D^{*}$ is (non)contextual
then can proceed in compliance with the general theory presented in
Section \ref{subsec: General-definition-of-contextuality}, and amounts
to determining if $\D^{*}$ has a \emph{maximally-connected} coupling
$S$ (see Remark (\ref{rem: multimaximally connected})). If no such
coupling exists, the computation of the degree of contextuality in
$\D^{*}$ can be done in compliance with Section \ref{subsec: Degree-of-contextuality}.

In a more general case, however, with an arbitrary number of rvs in
each connection, maximal couplings should be replaced with computing
what we call \emph{multimaximal couplings} \cite{DzhafarovKujala(2016)Fortschritte,DzhafarovKujala(2017)LNCS}.
\begin{defn}
A coupling $T_{q}$ of a connection $\D_{q}$ of a split representation
$\D$ is called \emph{multimaximal} if, for any $c,c'\in C$ such
that $c,c'\Yright q$, $\Pr\left[T_{q}^{c}=T_{q}^{c'}\right]$ is
maximal over all possible couplings of $\D_{q}$. (If the connection
contains two rvs, its multimaximal coupling is simply maximal.)
\end{defn}
A multimaximal coupling is known to have the following properties.
\begin{description}
\item [{Multimax1:}] The multimaximal coupling exists and is unique for
any connection $\D_{q}$ (\cite{DzhafarovKujala(2017)LNCS} Corollary
1).
\item [{Multmax2:}] $T_{q}$ is a multimaximal coupling of $\D_{q}$ if
and only if any subset of $T_{q}$ is a maximal coupling for the corresponding
subset of $\D_{q}$ (\cite{DzhafarovKujala(2017)LNCS} Theorem 5;
\cite{DzhafarovKujala(2016)Fortschritte} Theorem 2.3).
\item [{Multimax3:}] In a connection $\D_{q}$, if $\left\{ c_{1},\ldots,c_{n}\right\} $
is the set of all $c\Yright q$ enumerated so that 
\[
\Pr\left[D_{q}^{c_{1}}=1\right]\leq\ldots\leq\Pr\left[D_{q}^{c_{n}}=1\right],
\]
then $T_{q}$ is a multimaximal coupling of $\D_{q}$ if and only
if $\Pr\left[T_{q}^{c_{i}}=T_{q}^{c_{i+1}}\right]$ is maximal for
$i=1,\ldots,n-1$, over all possible couplings of $\D_{q}$ (\cite{DzhafarovKujala(2016)Fortschritte}
Theorem 2.3).
\end{description}

\section{\label{sec: A-systematic-study}The Largest Canonical Representation
of a Two-Element Connection}

We consider here the case when one is interested in all possible coarsenings
of the rvs in a system. The canonical/split representation of the
system then contains all splits of all rvs. We will investigate in
detail a fragment of the original (expanded) system involving just
two $k$-valued rvs within a single connection:
\begin{center}
\begin{tabular}{|c|c}
\cline{1-1} 
$R_{1}^{1}$$\ensuremath{}$ & $c=1$\tabularnewline
\cline{1-1} 
$R_{1}^{2}$$\ensuremath{}$ & $c=2$\tabularnewline
\cline{1-1} 
\multicolumn{1}{c}{$q=1$} & $\boxed{\R}$\tabularnewline
\end{tabular}
\par\end{center}

\noindent The canonical system with all splits of these $k$-valued
rvs is
\begin{center}
\begin{tabular}{c|c|c|c|c|c}
\cline{2-5} 
$D^{1}:$ & $D_{W1}^{1}$$ $ & $D_{W2}^{1}$$ $ & $ $ & $D_{W\left(2^{k-1}-1\right)}^{1}$$ $ & $c=1$\tabularnewline
\cline{2-5} 
$D^{2}:$ & $D_{W1}^{2}$$ $ & $D_{W2}^{2}$$ $ & $\cdots$ & $D_{W\left(2^{k-1}-1\right)}^{2}$$ $ & $c=2$\tabularnewline
\cline{2-5} 
\multicolumn{1}{c}{} & \multicolumn{1}{c}{$q=W1$} & \multicolumn{1}{c}{$W2$} & \multicolumn{1}{c}{$\cdots$} & \multicolumn{1}{c}{$W\left(2^{k-1}-1\right)$} & $\boxed{\D}$\tabularnewline
\end{tabular}
\par\end{center}

\noindent where $W1$, $W2$, etc. are the subsets $f^{-1}\left(1\right)$
chosen as explained in Section \ref{subsec: Dichotomizations-and-splits}
from the $2^{k-1}-1$ distinct dichotomizations $f$ of $\left\{ 1,\ldots,k\right\} $.
The number $2^{k-1}-1$ is arrived at by taking the number of all
subsets, subtracting $2$ improper subsets, and dividing by 2 because
one chooses only one of $W$ and $\left\{ 1,2,\ldots,k\right\} -W$.
The goal is to determine whether $\D$ is contextual. If it is, then
any canonical system that includes $\D$ as its subsystem (i.e., represents
an original system with $\R$ as part of its connection) is contextual. 

The two original rvs have distributions 
\begin{equation}
\Pr\left[R_{1}^{1}=i\right]=p_{i},\;\Pr\left[R_{1}^{2}=i\right]=q_{i}\;,i=1,2,\ldots,k.\label{eq: constraints for two bunches}
\end{equation}
A state (or value) of a bunch in the system $\D$ is a vector of $2^{k-1}-1$
zeroes and ones. However, the support of each of the bunches in system
$\D$ consists of at most $k$ corresponding states, and we can enumerate
them by any $k$ symbols, say, $1,2,\ldots,k$, as in the original
variable:
\begin{equation}
\Pr\left[D^{1}=i\right]=p_{i},\;\Pr\left[D^{2}=i\right]=q_{i}\;,i=1,2,\ldots,k,
\end{equation}

\noindent As a result, $\D=\left\{ D^{1},D^{2}\right\} $ has $k^{2}$
possible states that we can denote $ij$, with $i,j\in\left\{ 1,2,\ldots,k\right\} $.
A coupling $S=\left(S_{q}^{1},S_{q}^{2}\right)$ of $\D$ assigns
probabilities 
\begin{equation}
r_{ij}=\Pr\left[S_{q}^{1}=i,S_{q}^{2}=j\right],\;i,j\in\left\{ 1,\ldots,k\right\} ,\label{eq: def of rij}
\end{equation}
to these $k^{2}$ states so that they satisfy $2k$ linear constraints
imposed by (\ref{eq: constraints for two bunches}),
\begin{equation}
\sum_{j=1}^{k}r_{ij}=p_{i},\;\sum_{i=1}^{k}r_{ij}=q_{j},\;i,j\in\left\{ 1,\ldots,k\right\} .\label{eq: bunch 1 of 1-2}
\end{equation}
If $S$ is maximally-connected, then it should also satisfy $2^{k-1}-1$
linear constraints imposed by the maximal couplings of the corresponding
connections. Specifically, if $W=\left\{ i_{1},\ldots,i_{m}\right\} \subset\left\{ 1,\ldots,k\right\} $,
then the maximal coupling $\left(S_{W}^{1},S_{W}^{2}\right)$ of $\left(D_{W}^{1},D_{W}^{2}\right)$
is distributed as 
\begin{equation}
\left.\begin{array}{c}
\Pr\left[S_{W}^{1}=1\right]=\Pr\left[D_{W}^{1}=1\right]=p_{i_{1}}+p_{i_{2}}+\ldots+p_{i_{m}}\\
\Pr\left[S_{W}^{2}=1\right]=\Pr\left[D_{W}^{2}=1\right]=q_{i_{1}}+q_{i_{2}}+\ldots+q_{i_{m}}\\
\Pr\left[S_{W}^{1}=S_{W}^{2}=1\right]=\min\left(p_{i_{1}}+p_{i_{2}}+\ldots+p_{i_{m}},q_{i_{1}}+q_{i_{2}}+\ldots+q_{i_{m}}\right)
\end{array}\right].\label{eq: maximal for connections}
\end{equation}
Let us use the term $m$-split to designate any split $D_{W}$ with
an $m$-element set $W$ ($m\leq k/2$). Thus, $D_{W}$ with $W=\left\{ i\right\} $
is a 1-split, with $W=\left\{ i,j\right\} $ it is a 2-split, and
the higher-order splits appear beginning with $k>5$. Theorem \ref{thm: higher-order splits}
and its corollaries below show that in determining whether the system
$\D$ is contextual one needs to consider only the 1-splits and 2-splits.
Let us use the term \emph{1-2 system} for this subsystem of $\D$.
An overall coupling $S$ of $\D$ contains as its part a maximally-connected
coupling of the 1-2 system if and only if the probabilities $r_{ij}$
in (\ref{eq: def of rij}) satisfy (\ref{eq: maximal for connections})
for $m=1$ and $m=2$:
\begin{equation}
r_{ii}=\min\left(p_{i},q_{i}\right),\;i\in\left\{ 1,\ldots,k\right\} \label{eq: 1 of 1-2}
\end{equation}
and
\begin{equation}
r_{ii}+r_{ij}+r_{ji}+r_{jj}=\min\left(p_{i}+p_{j},q_{i}+q_{j}\right),\;i,j\in\left\{ 1,\ldots,k\right\} ,i<j.\label{eq: 2 of 1-2}
\end{equation}
That is, a maximally-connected coupling of the 1-2 system is described
by the $3k+\binom{k}{2}$ linear equations (\ref{eq: bunch 1 of 1-2})-(\ref{eq: 1 of 1-2})-(\ref{eq: 2 of 1-2}).
We have therefore the following necessary condition for noncontextuality
of $\D$.
\begin{thm}
\label{thm: 1-2 =000023 constraints}If the system $\D$ is noncontextual,
then the $3k+\binom{k}{2}$ linear equations (\ref{eq: bunch 1 of 1-2})-(\ref{eq: 1 of 1-2})-(\ref{eq: 2 of 1-2})
are satisfied. 
\end{thm}
\begin{rem}
\label{rem: Rank}Note that $3k+\binom{k}{2}<k^{2}$ for $k>5$. (For
completeness only, Theorem \ref{thm: rank} in the supplementary file
S shows that the rank of this system of equations is $2k-1+\binom{k}{2}$.)
\end{rem}
\begin{thm}
\label{thm: higher-order splits}In a maximally-connected coupling
$S$ of $\D$ with $k>5$, the distributions of the 1-splits and 2-splits
uniquely determine the probabilities of all higher-order splits. Specifically,
for any $2<m\leq k/2$, and any $W=\left\{ i_{1},\ldots,i_{m}\right\} \subset\left\{ 1,\ldots,k\right\} $,
the probability that the corresponding $m$-split equals 1 is
\begin{equation}
\begin{array}{l}
\min\left(p_{i_{1}}+p_{i_{2}}+\ldots+p_{i_{m}},q_{i_{1}}+q_{i_{2}}+\ldots+q_{i_{m}}\right)=\sum_{j=1}^{m}\min\left(p_{i_{j}},q_{i_{j}}\right)\\
+\sum_{j=1}^{m-1}\sum_{j'=j+1}^{m}\left[\min\left(p_{i_{j}}+p_{i_{j'}},q_{i_{j}}+q_{i_{j'}}\right)-\min\left(p_{i_{j}},q_{i_{j}}\right)-\min\left(p_{i_{j'}},q_{i_{j'}}\right)\right].
\end{array}\label{eq: relation}
\end{equation}
\end{thm}
It is easy to find numerical examples of the distributions of $R_{1}^{1}$
and $R_{1}^{2}$ for which (\ref{eq: relation}) is violated (see
Example \ref{exa: Relation may be violated} in the supplementary
file S). As shown below, however, (\ref{eq: relation}) cannot be
violated if a maximally-connected coupling for the 1-2 system exists.
It follows from the fact that the statement of Theorem \ref{thm: 1-2 =000023 constraints}
can be reversed: (\ref{eq: bunch 1 of 1-2})-(\ref{eq: 1 of 1-2})-(\ref{eq: 2 of 1-2})
imply that $\D$ is noncontextual. We establish this fact by first
characterizing the distributions of $R_{1}^{1}$ and $R_{1}^{2}$
for a noncontextual 1-2 system (Theorem \ref{thm: 1-2 system} with
Corollary \ref{cor: 1-2-system}), and then showing that (\ref{eq: relation})
always holds for such distributions (Theorem \ref{thm: exists unique}).
\begin{thm}
\label{thm: 1-2 system}A maximally-connected coupling for a 1-2 system
is unique if it exists. In this coupling, the only pairs of $ij$
in (\ref{eq: def of rij}) that may have nonzero probabilities assigned
to them are the diagonal states $\left\{ 11,22,\ldots,kk\right\} $
and either the states $\left\{ i1,i2,\ldots,ik\right\} $ for a single
fixed $i$ or the states $\left\{ 1j,2j,\ldots,kj\right\} $ for a
single fixed $j$ ($i,j=1,\ldots,k$).
\end{thm}
Assuming, with no loss of generality, that the single fixed $i$ or
the single fixed $j$ in the formulation above is $2,$ the theorem
says that the nonzero probabilities assigned to the states of the
maximally-connected coupling (shown below for $k=4$) could only occupy
the cells marked with asterisks:
\begin{center}
\begin{tabular}{c|c|c|c|c|}
\multicolumn{1}{c}{} & \multicolumn{1}{c}{$1$} & \multicolumn{1}{c}{$2$} & \multicolumn{1}{c}{$3$} & \multicolumn{1}{c}{$4$}\tabularnewline
\cline{2-5} 
$1$ & $\ast$ & $\ast$ & $0$ & $0$\tabularnewline
\cline{2-5} 
$2$ & $0$ & $\ast$ & $0$ & $0$\tabularnewline
\cline{2-5} 
$3$ & $0$ & $\ast$ & $\ast$ & $0$\tabularnewline
\cline{2-5} 
$4$ & $0$ & $\ast$ & $0$ & $\ast$\tabularnewline
\cline{2-5} 
\end{tabular}$\quad$or %
\begin{tabular}{c|c|c|c|c|}
\multicolumn{1}{c}{} & \multicolumn{1}{c}{$1$} & \multicolumn{1}{c}{$2$} & \multicolumn{1}{c}{$3$} & \multicolumn{1}{c}{$4$}\tabularnewline
\cline{2-5} 
$1$ & $\ast$ & $0$ & $0$ & $0$\tabularnewline
\cline{2-5} 
$2$ & $\ast$ & $\ast$ & $\ast$ & $\ast$\tabularnewline
\cline{2-5} 
$3$ & $0$ & $0$ & $\ast$ & $0$\tabularnewline
\cline{2-5} 
$4$ & $0$ & $0$ & $0$ & $\ast$\tabularnewline
\cline{2-5} 
\end{tabular}.
\par\end{center}
\begin{cor}
\label{cor: 1-2-system}The 1-2 system for the original rvs $R_{1}^{1},R_{1}^{2}$
has a maximally-connected coupling if and only if either $p_{i}>q_{i}$
for no more than one $i$ (this single possible $i$ being the single
fixed $i$ in the formulation of the theorem), or $p_{j}<q_{j}$ for
no more than one $j$ (this single possible $j$ being the single
fixed $j$ in the formulation of the theorem), $i,j\in\left\{ 1,\ldots,k\right\} $.
\end{cor}
The relationship between $\left(p_{1},\ldots,p_{k}\right)$ and $\left(q_{1},\ldots,q_{k}\right)$
described in this corollary is some form of stochastic dominance for
categorical rvs, but it does not seem to have been previously identified.
We propose to say that $R_{1}^{1}$\emph{ nominally dominates }$R_{1}^{2}$
if $p_{i}<q_{i}$ for no more than one value of $i=1,\ldots,k$ (i.e.,
$p_{i}\geq q_{i}$ for at least $k-1$ of them). Two categorical rvs
nominally dominate each other if and only if either they are identically
distributed or $k=2$. Using this notion, and combining Corollary
\ref{cor: 1-2-system} with Theorems \ref{thm: 1-2 =000023 constraints}
and \ref{thm: 1-2 system}, we get the main result of this section.
\begin{thm}
\label{thm: exists unique}The system $\D$ is noncontextual if and
only if its 1-2 subsystem is noncontextual, i.e., if and only if one
of the $R_{1}^{1}$\emph{ and }$R_{1}^{2}$ nominally dominates the
other.
\end{thm}

\section{Concluding remarks}

Contextuality analysis of an empirical situation involves the following
sequence of steps:

\[
\xymatrix@C=1cm{\mbox{\ensuremath{\begin{array}{c}
 \textnormal{empirical}\\
 \textnormal{measurements} 
\end{array}}\ar[r]} & \textnormal{\ensuremath{\begin{array}{c}
 \textnormal{initial}\\
 \textnormal{system of rvs} 
\end{array}}}\ar[r] & \textnormal{\ensuremath{\begin{array}{c}
 \textnormal{expanded}\\
 \textnormal{system of rvs} 
\end{array}}}\ar[r] & \textnormal{\ensuremath{\begin{array}{c}
 \textnormal{canonical/split}\\
 \textnormal{representation} 
\end{array}}}}
\]
In the initial system, measurements are represented by rvs each of
which generally has multiple values. Expansion means adding to the
system new conteNts with corresponding connections (conteNt-sharing
rvs) computed as functions of the existing connections. In a canonical
representation of the system all rvs are binary, and the connections
are coupled multimaximally, meaning essentially that one deals with
their elements pairwise. The issue of contextuality is reduced to
that of compatibility of the unique couplings for pairs of conteNt-sharing
rvs with the known distributions of the conteXt-sharing bunches of
rvs. Coupling the connections multimaximally ensures that a noncontextual
system has all its subsystems noncontextual too.

The canonical system of rvs is uniquely determined by the expanded
system, but the latter is inherently non-unique, it depends on what
aspects of the empirical situation one wishes to include in the system.
Thus, it is one's choice rather than a general rule whether one considers
a multi-valued measurement as representable by \emph{all} or only
\emph{some} of its possible coarsenings. If one chooses all coarsenings,
the split/canonical representation involves all dichotomizations,
and then Theorem \ref{thm: exists unique} says that the canonical
system is noncontextual only if, for any pair of rvs $R_{q}^{c},R_{q}^{c'}$
in the expanded system, one of them, say $R_{q}^{c}$, ``nominally
dominates'' the other. This domination means that $\Pr\left[R_{q}^{c}=x\right]<\Pr\left[R_{q}^{c'}=x\right]$
holds for no more than one value $x$ of these rvs: a stringent necessary
condition for noncontextuality, likely to be violated in many empirical
systems. 

This is of special interest for contextuality studies outside quantum
physics. Historically, the search for non-quantum contextual systems
was motivated by the possibility of applying quantum-theoretic formalisms
in such fields as biology \cite{Asanoetal2015}, psychology \cite{WangBusemeyer2013,BusemeyerBruza2012},
economics \cite{HavenKhrennikov2012}, and political science \cite{Khrennikova}.
In CbD, the notion of contextuality is not tied to quantum formalisms
in any special way. The possibility of non-quantum contextual systems
here is motivated by treating contextuality as an abstract probabilistic
issue: there are no a priori reasons why a system of rvs describing,
say, human behavior could not be contextual if it is qualitatively
(i.e., up to specific probability values) the same as a contextual
one describing particle spins. Nevertheless, all known to us systems
with dichotomous responses investigated for potential contextuality
(with the exception of one, very recent experiment) have been found
to be noncontextual \cite{DzhafarovKujalaCervantesZhangJones(2016),DzhafarovZhangKujala(2015)Isthere,CervantesDzhafarov2017}.
The use of canonical representations with dichotomizations of multiple-choice
responses offers new possibilities. 

In some cases, however, the use of all possible dichotomizations is
not justifiable. Notably, if the values of an rv are linearly ordered,
$x_{1}<x_{2}<\ldots,x_{N}$, it may be natural to only allow dichotomizations
$f$ with $f^{-1}\left(1\right)$ containing several successive values,
$\left\{ x_{l},x_{l+1},\ldots,x_{L}\right\} $, for some $l,L\in\left\{ 1,\ldots,N\right\} $.
An even stronger restriction would be to only allow ``cuts,'' with
$f^{-1}\left(1\right)=\left\{ x_{l},x_{l+1},\ldots,x_{N}\right\} $
or $\left\{ x_{1},x_{2},\ldots,x_{l-1}\right\} $. 
\begin{center}
\includegraphics[bb=0bp 100bp 792bp 550bp,scale=0.4]{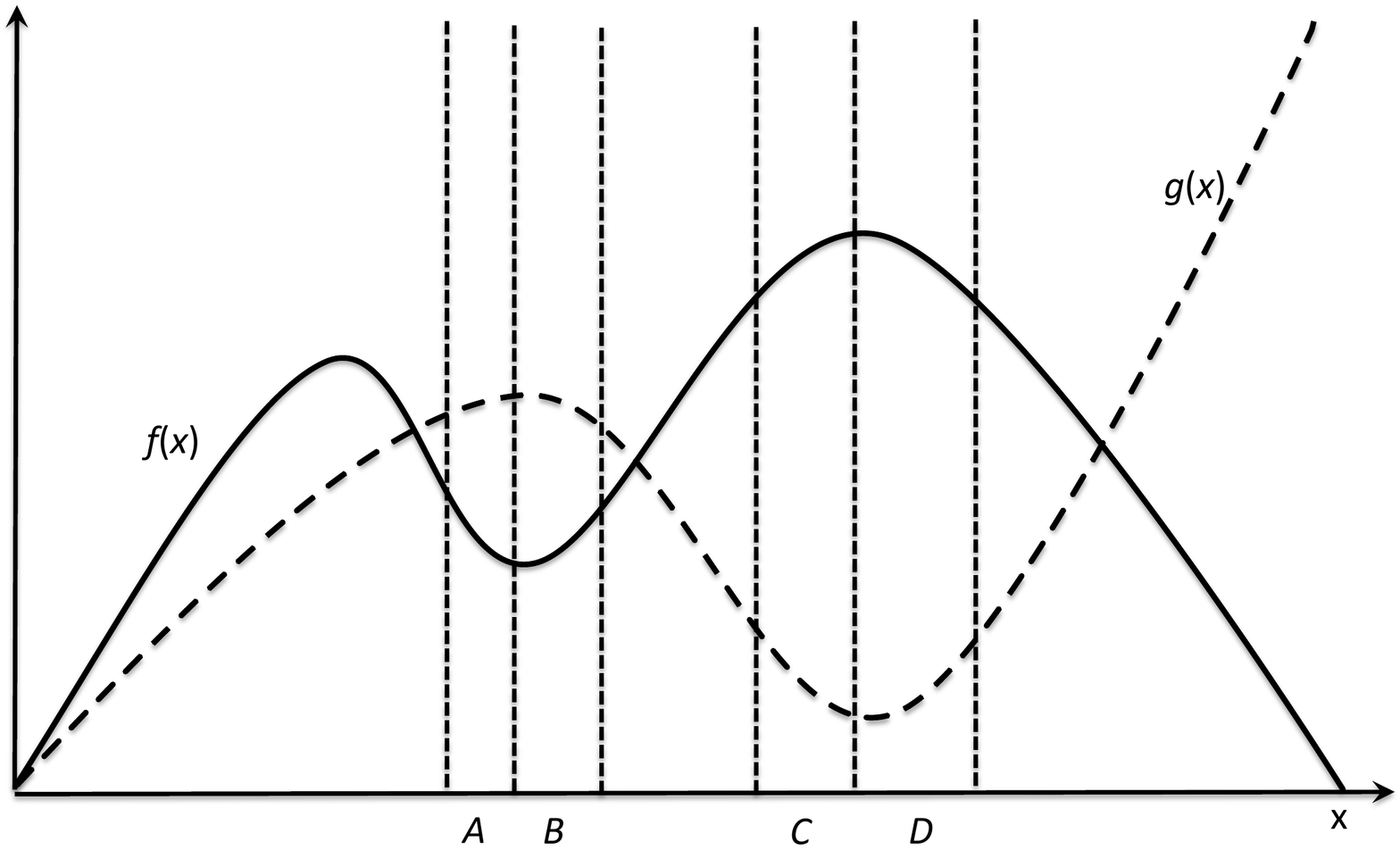}
\par\end{center}

\noindent Stronger restrictions on possible dichotomizations translate
into stronger restrictions on the pairs $R_{q}^{c},R_{q'}^{c}$ whose
canonical representation is contextual. This fact is especially important
if one considers expanding CbD beyond categorical rvs. Thus, it is
easy to see that if one considers all possible dichotomizations of
two conteNt-sharing rvs with continuous densities on the set of real
numbers, then the system will be contextual whenever the two distributions
are not identical. Let the densities of these rvs be $f\left(x\right)$
and $g\left(x\right)$ shown in the graphic above. If the set of all
splits of these rvs forms a noncontextual system, then any discretization
of these rvs should satisfy Corollary \ref{cor: 1-2-system} to Theorem
\ref{thm: 1-2 system}. That is, for any $k>2$ and any partition
$H_{1},\ldots,H_{k}$ of the set of reals into intervals, we should
have either
\begin{equation}
\begin{array}{c}
\int_{H_{i}}f\left(x\right)dx<\int_{H_{i}}g\left(x\right)dx\textnormal{ for no more than one of \ensuremath{i=1,\ldots,k},}\\
\textnormal{or}\\
\int_{H_{i}}f\left(x\right)dx>\int_{H_{i}}g\left(x\right)dx\textnormal{ for no more than one of \ensuremath{i=1,\ldots,k}}.
\end{array}\label{eq: continuous}
\end{equation}
This is, however, impossible unless $f\left(x\right)=g\left(x\right)$.
If they are different, then $f$ exceeds $g$ on some interval, and
$g$ exceeds $f$ on some other interval. If we take any two subintervals
within each of these intervals (in the graphic they are denoted by
$A,B$ and $C,D$), any partition $H_{1},\ldots,H_{k}$ that includes
$A,B,C,D$ will violate (\ref{eq: continuous}). The development of
the theory of canonical representations with variously restricted
sets of splits is a task for future work.

\paragraph*{Data accessibility. }

See Remark \ref{rem: The-proofs-of}. 

\paragraph*{Competing interests. }

We have no financial or non-financial competing interests.

\paragraph*{Authors\textquoteright{} contributions. }

All authors significantly contributed to the development of the theory
and drafting of the paper.

\paragraph*{Acknowledgments.}

We have greatly benefited from discussions with Matt Jones, Samson
Abramsky, Rui Soares Barbosa, and Pawel Kurzynski. 

\paragraph*{Funding statement. }

This research has been supported by AFOSR grant FA9550-14-1-0318.

\newpage{}

\appendix
\setcounter{page}{1}

\renewcommand{\thesection}{\Alph{section}}

\renewcommand{\thesection}{S}

\numberwithin{equation}{section}

\section{Supplementary Text to ``Contextuality in Canonical Systems of Random
Variables'' by Ehtibar N. Dzhafarov, Víctor H. Cervantes, and Janne
V. Kujala (Phil. Trans. Roy. Soc. A xxx, 10.1098/rsta.2016.0389)}
\begin{thm}[Section \ref{sec: A-systematic-study}, Remark \ref{rem: Rank}]
\label{thm: rank}The rank of the system of linear equations (\ref{eq: bunch 1 of 1-2})-(\ref{eq: 1 of 1-2})-(\ref{eq: 2 of 1-2})
is \textup{$2k-1+\binom{k}{2}$.}
\end{thm}
\begin{proof}[Proof of Theorem \ref{thm: rank}]
This system of linear equations can be written as
\[
\mathbf{M}\times\mathbf{X}=\mathbf{P},
\]
where
\[
\mathbf{P}^{T}=\left(\begin{array}{l}
\overset{k}{\overbrace{p_{1},\ldots,p_{k}}},\overset{k}{\overbrace{q_{1},\ldots,q_{k}}},\overset{k}{\overbrace{\min\left(p_{1},q_{1}\right),\ldots,\min\left(p_{k},q_{k}\right)}},\\
\\
\\
\overset{\binom{k}{2}}{\overbrace{\min\left(p_{1}+p_{2},q_{1}+q_{2}\right),\ldots,\min\left(p_{k-1}+p_{k},q_{k-1}+q_{k}\right)}}
\end{array}\right),
\]
\[
\mathbf{X}^{T}=\left\{ x_{ij}:i,j\in\left\{ 1,\ldots,k\right\} \right\} ,
\]
and $\mathbf{M}$ is a Boolean matrix. The $\left(k+k+k+\binom{k}{2}\right)$
rows of matrix \textbf{$\mathbf{M}$} correspond to the elements of
$\mathbf{P}$ and can be labeled as 
\[
\left(\overset{k}{\overbrace{\mathbf{r}_{1\cdot},\ldots,\mathbf{r}{}_{k\cdot}}},\overset{k}{\overbrace{\mathbf{r}_{\cdot1},\ldots,\mathbf{r}_{\cdot k}}},\overset{k}{\overbrace{\mathbf{r}_{11},\ldots,\mathbf{r}_{kk}}},\overset{\binom{k}{2}}{\overbrace{\mathbf{r}_{12},\ldots,\mathbf{r}_{k-1,k}}}\right),
\]
whereas the $k^{2}$ columns of $\mathbf{M}$ correspond to the elements
of $\mathbf{X}$ and can be labeled as 
\[
\left\{ \mathbf{c}_{ij}:i,j\in\left\{ 1,\ldots,k\right\} \right\} .
\]
Thus, if $k=4$, the matrix $\mathbf{M}$ is 
\[
\begin{array}{ccccccccccccccccc}
\begin{array}{cc}
 & \mathbf{c}\\
\mathbf{r}
\end{array} & 11 & 12 & 13 & 14 & 21 & 22 & 23 & 24 & 31 & 32 & 33 & 34 & 41 & 42 & 43 & 44\\
1\cdot & 1 & 1 & 1 & 1\\
2\cdot &  &  &  &  & 1 & 1 & 1 & 1\\
3\cdot &  &  &  &  &  &  &  &  & 1 & 1 & 1 & 1\\
4\cdot &  &  &  &  &  &  &  &  &  &  &  &  & 1 & 1 & 1 & 1\\
\cdot1 & 1 &  &  &  & 1 &  &  &  & 1 &  &  &  & 1\\
\cdot2 &  & 1 &  &  &  & 1 &  &  &  & 1 &  &  &  & 1\\
\cdot3 &  &  & 1 &  &  &  & 1 &  &  &  & 1 &  &  &  & 1\\
\cdot4 &  &  &  & 1 &  &  &  & 1 &  &  &  & 1 &  &  &  & 1\\
11 & 1\\
22 &  &  &  &  &  & 1\\
33 &  &  &  &  &  &  &  &  &  &  & 1\\
44 &  &  &  &  &  &  &  &  &  &  &  &  &  &  &  & 1\\
12 & 1 & 1 &  &  & 1 & 1\\
13 & 1 &  & 1 &  &  &  &  &  & 1 &  & 1\\
14 & 1 &  &  & 1 &  &  &  &  &  &  &  &  & 1 &  &  & 1\\
23 &  &  &  &  &  & 1 & 1 &  &  & 1 & 1\\
24 &  &  &  &  &  & 1 &  & 1 &  &  &  &  &  & 1 &  & 1\\
34 &  &  &  &  &  &  &  &  &  &  & 1 & 1 &  &  & 1 & 1
\end{array}
\]
We will continue to illustrate the steps of the proof using this matrix.
We begin by adding to $\mathbf{M}$ the row $\mathbf{r}_{all}$ with
all cells equal to 1, and denote the new matrix $\mathbf{M'}$. 
\[
\begin{array}{ccccccccccccccccc}
\begin{array}{cc}
 & \mathbf{c}\\
\mathbf{r}
\end{array} & 11 & 12 & 13 & 14 & 21 & 22 & 23 & 24 & 31 & 32 & 33 & 34 & 41 & 42 & 43 & 44\\
1\cdot & 1 & 1 & 1 & 1\\
2\cdot &  &  &  &  & 1 & 1 & 1 & 1\\
3\cdot &  &  &  &  &  &  &  &  & 1 & 1 & 1 & 1\\
4\cdot &  &  &  &  &  &  &  &  &  &  &  &  & 1 & 1 & 1 & 1\\
\cdot1 & 1 &  &  &  & 1 &  &  &  & 1 &  &  &  & 1\\
\cdot2 &  & 1 &  &  &  & 1 &  &  &  & 1 &  &  &  & 1\\
\cdot3 &  &  & 1 &  &  &  & 1 &  &  &  & 1 &  &  &  & 1\\
\cdot4 &  &  &  & 1 &  &  &  & 1 &  &  &  & 1 &  &  &  & 1\\
11 & 1\\
22 &  &  &  &  &  & 1\\
33 &  &  &  &  &  &  &  &  &  &  & 1\\
44 &  &  &  &  &  &  &  &  &  &  &  &  &  &  &  & 1\\
12 & 1 & 1 &  &  & 1 & 1\\
13 & 1 &  & 1 &  &  &  &  &  & 1 &  & 1\\
14 & 1 &  &  & 1 &  &  &  &  &  &  &  &  & 1 &  &  & 1\\
23 &  &  &  &  &  & 1 & 1 &  &  & 1 & 1\\
24 &  &  &  &  &  & 1 &  & 1 &  &  &  &  &  & 1 &  & 1\\
34 &  &  &  &  &  &  &  &  &  &  & 1 & 1 &  &  & 1 & 1\\
all & 1 & 1 & 1 & 1 & 1 & 1 & 1 & 1 & 1 & 1 & 1 & 1 & 1 & 1 & 1 & 1
\end{array}
\]
This does not change the rank of the matrix since $\mathbf{r}_{all}$
is the sum of all $\mathbf{r}_{\cdot i}$. Then we observe that the
rows $\mathbf{r}_{k\cdot}$, $\mathbf{r}_{\cdot k}$, and all $\mathbf{r}_{ik}$
with $i<k$ can be deleted as they are linear combinations of the
remaining rows of $\mathbf{M'}$. Indeed, it can be checked directly
that
\[
\mathbf{r}_{k\cdot}=\mathbf{r}_{all}-\sum_{i=1}^{k-1}\mathbf{r}_{i\cdot},
\]
\[
\mathbf{r}_{\cdot k}=\mathbf{r}_{all}-\sum_{i=1}^{k-1}\mathbf{r}_{\cdot i},
\]
\[
\left(\mathbf{r}_{ik}-\mathbf{r}_{ii}-\mathbf{r}_{kk}\right)=\left(\mathbf{r}_{i\cdot}-\mathbf{r}_{ii}\right)+\left(\mathbf{r}_{\cdot i}-\mathbf{r}_{ii}\right)-\sum_{l<i}\left(\mathbf{r}_{li}-\mathbf{r}_{ll}-\mathbf{r}_{ii}\right)-\sum_{l>i}^{l<k}\left(\mathbf{r}_{il}-\mathbf{r}_{ii}-\mathbf{r}_{ll}\right),
\]
for all $i<k$. Moreover, one can also delete $\mathbf{r}_{kk}$,
because
\[
\sum_{i<j<k}\left(\mathbf{r}_{ij}-\mathbf{r}_{ii}-\mathbf{r}_{jj}\right)+\sum_{i<k}\left(\mathbf{r}_{ik}-\mathbf{r}_{ii}-\mathbf{r}_{kk}\right)+\sum_{i<k}\mathbf{r}_{ii}+\mathbf{r}_{kk}=\mathbf{r}_{all}.
\]
Let the resulting matrix be $\mathbf{M''}$:
\[
\begin{array}{ccccccccccccccccc}
\begin{array}{cc}
 & \mathbf{c}\\
\mathbf{r}
\end{array} & 11 & 12 & 13 & 14 & 21 & 22 & 23 & 24 & 31 & 32 & 33 & 34 & 41 & 42 & 43 & 44\\
1\cdot & 1 & 1 & 1 & 1\\
2\cdot &  &  &  &  & 1 & 1 & 1 & 1\\
3\cdot &  &  &  &  &  &  &  &  & 1 & 1 & 1 & 1\\
\cdot1 & 1 &  &  &  & 1 &  &  &  & 1 &  &  &  & 1\\
\cdot2 &  & 1 &  &  &  & 1 &  &  &  & 1 &  &  &  & 1\\
\cdot3 &  &  & 1 &  &  &  & 1 &  &  &  & 1 &  &  &  & 1\\
11 & 1\\
22 &  &  &  &  &  & 1\\
33 &  &  &  &  &  &  &  &  &  &  & 1\\
12 & 1 & 1 &  &  & 1 & 1\\
13 & 1 &  & 1 &  &  &  &  &  & 1 &  & 1\\
23 &  &  &  &  &  & 1 & 1 &  &  & 1 & 1\\
all & 1 & 1 & 1 & 1 & 1 & 1 & 1 & 1 & 1 & 1 & 1 & 1 & 1 & 1 & 1 & 1
\end{array}
\]

This matrix contains 
\[
\overset{initial}{3k+\binom{k}{2}}-\underset{\mathbf{r}_{k\cdot},\mathbf{r}_{\cdot k},\mathbf{r}_{kk}}{\underbrace{3}}-\overset{all\:\mathbf{r}_{ik},i<k}{\overbrace{\left(k-1\right)}}+\underset{\mathbf{r}_{all}}{\underbrace{1}}=2k-1+\binom{k}{2}
\]
rows. We prove that this matrix is of full row rank. Consider equation
\[
\sum_{all\:\mathbf{r}\textnormal{ in }\mathbf{M}''}\alpha_{\mathbf{r}}\mathbf{r}=0.
\]
We use the following principle: if a row $\mathbf{r}$ intersects
a columns whose only nonzero entry is in the row $\mathbf{r}$, then
$\alpha_{\mathbf{r}}=0$, and we can delete the row $\mathbf{r}$
from the matrix, decreasing the row rank of the matrix by 1. The following
statements can be directly verified.

$\mathbf{r}_{all}$ can be deleted because column $\mathbf{c}_{kk}$
has its only 1 in $\mathbf{r}_{all}$. 

\[
\begin{array}{ccccccccccccccccc}
\begin{array}{cc}
 & \mathbf{c}\\
\mathbf{r}
\end{array} & 11 & 12 & 13 & 14 & 21 & 22 & 23 & 24 & 31 & 32 & 33 & 34 & 41 & 42 & 43 & 44\\
1\cdot & 1 & 1 & 1 & 1\\
2\cdot &  &  &  &  & 1 & 1 & 1 & 1\\
3\cdot &  &  &  &  &  &  &  &  & 1 & 1 & 1 & 1\\
\cdot1 & 1 &  &  &  & 1 &  &  &  & 1 &  &  &  & 1\\
\cdot2 &  & 1 &  &  &  & 1 &  &  &  & 1 &  &  &  & 1\\
\cdot3 &  &  & 1 &  &  &  & 1 &  &  &  & 1 &  &  &  & 1\\
11 & 1\\
22 &  &  &  &  &  & 1\\
33 &  &  &  &  &  &  &  &  &  &  & 1\\
12 & 1 & 1 &  &  & 1 & 1\\
13 & 1 &  & 1 &  &  &  &  &  & 1 &  & 1\\
23 &  &  &  &  &  & 1 & 1 &  &  & 1 & 1
\end{array}
\]

Then each of $\mathbf{r}_{\cdot i}$can be deleted because the column
$\mathbf{c}_{ki}$ has its only 1 in $\mathbf{r}_{\cdot i}$ ($i=1,\ldots,k-1$).
\[
\begin{array}{ccccccccccccccccc}
\begin{array}{cc}
 & \mathbf{c}\\
\mathbf{r}
\end{array} & 11 & 12 & 13 & 14 & 21 & 22 & 23 & 24 & 31 & 32 & 33 & 34 & 41 & 42 & 43 & 44\\
1\cdot & 1 & 1 & 1 & 1\\
2\cdot &  &  &  &  & 1 & 1 & 1 & 1\\
3\cdot &  &  &  &  &  &  &  &  & 1 & 1 & 1 & 1\\
11 & 1\\
22 &  &  &  &  &  & 1\\
33 &  &  &  &  &  &  &  &  &  &  & 1\\
12 & 1 & 1 &  &  & 1 & 1\\
13 & 1 &  & 1 &  &  &  &  &  & 1 &  & 1\\
23 &  &  &  &  &  & 1 & 1 &  &  & 1 & 1
\end{array}
\]

Then each of $\mathbf{r}_{i\cdot}$ can be deleted because the column
$\mathbf{c}_{ik}$ has its only 1 in $\mathbf{r}_{i\cdot}$ ($i=1,\ldots,k-1$).

\[
\begin{array}{ccccccccccccccccc}
\begin{array}{cc}
 & \mathbf{c}\\
\mathbf{r}
\end{array} & 11 & 12 & 13 & 14 & 21 & 22 & 23 & 24 & 31 & 32 & 33 & 34 & 41 & 42 & 43 & 44\\
11 & 1\\
22 &  &  &  &  &  & 1\\
33 &  &  &  &  &  &  &  &  &  &  & 1\\
12 & 1 & 1 &  &  & 1 & 1\\
13 & 1 &  & 1 &  &  &  &  &  & 1 &  & 1\\
23 &  &  &  &  &  & 1 & 1 &  &  & 1 & 1
\end{array}
\]

Then each of $\mathbf{r}_{ij}$ can be deleted because the column
$\mathbf{c}_{ji}$ has its only 1 in $\mathbf{r}_{ij}$ ($i,j\in\left\{ 1,\ldots,k-1\right\} ,i<j$).

\[
\begin{array}{ccccccccccccccccc}
\begin{array}{cc}
 & \mathbf{c}\\
\mathbf{r}
\end{array} & 11 & 12 & 13 & 14 & 21 & 22 & 23 & 24 & 31 & 32 & 33 & 34 & 41 & 42 & 43 & 44\\
11 & 1\\
22 &  &  &  &  &  & 1\\
33 &  &  &  &  &  &  &  &  &  &  & 1
\end{array}
\]

This leaves only $\mathbf{r}_{11},\ldots,\mathbf{r}_{\left(k-1\right)\left(k-1\right)}$
that are obviously linearly independent.
\end{proof}
\rule[0.5ex]{0.9\columnwidth}{1pt}
\begin{thm*}[Section \ref{sec: A-systematic-study}, Theorem \ref{thm: higher-order splits}]
In a maximally-connected coupling $S$ of $\D$ with $k>5$, the
distributions of the 1-splits and 2-splits uniquely determine the
probabilities of all higher-order splits. Specifically, for any $2<m\leq k/2$,
and any $W=\left\{ i_{1},\ldots,i_{m}\right\} \subset\left\{ 1,\ldots,k\right\} $,
the probability that the corresponding $m$-split equals 1 is
\begin{equation}
\begin{array}{l}
\min\left(p_{i_{1}}+p_{i_{2}}+\ldots+p_{i_{m}},q_{i_{1}}+q_{i_{2}}+\ldots+q_{i_{m}}\right)=\sum_{j=1}^{m}\min\left(p_{i_{j}},q_{i_{j}}\right)\\
+\sum_{j=1}^{m-1}\sum_{j'=j+1}^{m}\left[\min\left(p_{i_{j}}+p_{i_{j'}},q_{i_{j}}+q_{i_{j'}}\right)-\min\left(p_{i_{j}},q_{i_{j}}\right)-\min\left(p_{i_{j'}},q_{i_{j'}}\right)\right].
\end{array}\label{eq: relation-1}
\end{equation}
\end{thm*}
\begin{proof}[Proof of Theorem \ref{thm: higher-order splits}]
 From (\ref{eq: 1 of 1-2}) and (\ref{eq: 2 of 1-2}),
\[
\begin{array}{cccc}
r_{12}+r_{21} & = & \min\left(p_{1}+p_{2},q_{1}+q_{2}\right)-\min\left(p_{1},q_{1}\right)-\min\left(p_{2},q_{2}\right)\\
\vdots & \vdots & \vdots\\
r_{ij}+r_{ji} & = & \min\left(p_{i}+p_{j},q_{i}+q_{j}\right)-\min\left(p_{i},q_{i}\right)-\min\left(p_{j},q_{j}\right) & (i<j).\\
\vdots & \vdots & \vdots\\
r_{\left(k-1\right)k}+r_{k\left(k-1\right)} & = & \min\left(p_{k-1}+p_{k},q_{k-1}+q_{k}\right)-\min\left(p_{k-1},q_{k-1}\right)-\min\left(p_{k},q_{k}\right)
\end{array}
\]
Consider an $m$-split with $2<m\leq k/2$, and assume without loss
of generality that $W=\left(1,\ldots,m\right)$. We have 
\begin{equation}
\sum_{i=1}^{m}\sum_{j=1}^{m}r_{ij}=\min\left(p_{1}+\ldots+p_{m},q_{1}+\ldots+q_{m}\right).
\end{equation}
The left-hand-side sum can be presented as
\[
\begin{array}{l}
\sum_{i=1}^{m}r_{ii}+\sum_{i=1}^{m-1}\sum_{j=i+1}^{m}\left(r_{ij}+r_{ji}\right)\\
\\
=\sum_{i=1}^{m}\min\left(p_{i},q_{i}\right)+\sum_{i=1}^{m-1}\sum_{j=i+1}^{m}\left[\min\left(p_{i}+p_{j},q_{i}+q_{j}\right)-\min\left(p_{i},q_{i}\right)-\min\left(p_{j},q_{j}\right)\right],
\end{array}
\]
whence we get (\ref{eq: relation}).
\end{proof}
\rule[0.5ex]{0.9\columnwidth}{1pt}
\begin{example}[showing that the relation (\ref{eq: relation}) may be violated, see
Section \ref{sec: A-systematic-study}.]
\label{exa: Relation may be violated} If

\begin{center}

\begin{tabular}{c|c|c|c|c|c|c|}
\cline{2-7} 
$R_{1}^{1}$= & $1$ & $2$ & $3$ & $4$ & $0$ & $0$\tabularnewline
\cline{2-7} 
prob. mass $p=$ & $.6$ & $.1$ & $.1$ & $.2$ & $0$ & $0$\tabularnewline
\cline{2-7} 
\end{tabular},$\quad$%
\begin{tabular}{c|c|c|c|c|c|c|}
\cline{2-7} 
$R_{1}^{2}$= & $1$ & $2$ & $3$ & $4$ & $0$ & $0$\tabularnewline
\cline{2-7} 
prob. mass $q=$ & $.2$ & $.3$ & $.4$ & $.1$ & $0$ & $0$\tabularnewline
\cline{2-7} 
\end{tabular},

\end{center}

\noindent then
\[
\begin{array}{l}
\overset{.8}{\overbrace{\min\left(p_{1}+p_{2}+p_{3},q_{1}+q_{2}+q_{3}\right)}}\\
\\
\not=\left.\begin{array}{rll}
\min\left(p_{1},q_{1}\right) &  & .2\\
+\min\left(p_{2},q_{2}\right) &  & .1\\
+\min\left(p_{3},q_{3}\right) &  & .1\\
+\min\left(p_{1}+p_{2},q_{1}+q_{2}\right)-\min\left(p_{1},q_{1}\right)-\min\left(p_{2},q_{2}\right) &  & .5-.2-.1\\
+\min\left(p_{1}+p_{3},q_{1}+q_{3}\right)-\min\left(p_{1},q_{1}\right)-\min\left(p_{3},q_{3}\right) &  & .6-.2-.1\\
+\min\left(p_{2}+p_{3},q_{2}+q_{3}\right)-\min\left(p_{2},q_{2}\right)-\min\left(p_{3},q_{3}\right) &  & .2-.1-.1
\end{array}\right\} =.5\\
\\
\end{array}\hfill\square
\]
\end{example}
\rule[0.5ex]{0.9\columnwidth}{1pt}
\begin{thm*}[Section \ref{sec: A-systematic-study}, Theorem \ref{thm: 1-2 system}]
A maximally-connected coupling for a 1-2 system is unique if it exists.
In this coupling, the only pairs of $ij$ in (\ref{eq: def of rij})
that may have nonzero probabilities assigned to them are the diagonal
states $\left\{ 11,22,\ldots,kk\right\} $ and either the states $\left\{ i1,i2,\ldots,ik\right\} $
for a single fixed $i$ or the states $\left\{ 1j,2j,\ldots,kj\right\} $
for a single fixed $j$ ($i,j=1,\ldots,k$).
\end{thm*}
\begin{proof}[Proof of Theorem \ref{thm: 1-2 system}]
(The matrices illustrating the proof are shown for $k>6$ but the
theorem is valid for all $k>1$.) If the only nonzero entries in the
matrix are in the main diagonal, the theorem is trivially true. Assume
therefore that $r_{ij}>0$ for some $i\not=j$. Without loss of generality,
we can assume that $r_{12}>0$ and $p_{1}+p_{2}\leq q_{1}+q_{2}$.
Indeed, if some $r_{ij}>0$, we can always rename the values so that
$i=1$ and $j=2$; and if $p_{1}+p_{2}>q_{1}+q_{2}$, then we can
simply rename all $p$s into $q$s and vice versa. In the following
we will use the expression ``$r_{ij}$ is $p$-minimized'' if $p_{i}+p_{j}\leq q_{i}+q_{j}$,
and ``$r_{ij}$ is $q$-minimized'' if $p_{i}+p_{j}\geq q_{i}+q_{j}$
(in both cases, $i\not=j$).

We have (the empty cells are those whose value is to be determined
later)

\begin{center}

\begin{tabular}{c|c|c|c|c|c|c|c||c}
\multicolumn{1}{c}{} & \multicolumn{1}{c}{$1$} & \multicolumn{1}{c}{$2$} & \multicolumn{1}{c}{$3$} & \multicolumn{1}{c}{$4$} & \multicolumn{1}{c}{$5$} & \multicolumn{1}{c}{$6$} & \multicolumn{1}{c}{$\ldots$} & \tabularnewline
\cline{2-9} 
$1$ & $r_{11}$ & $r_{12}>0$ &  &  &  &  &  & $p_{1}$\tabularnewline
\cline{2-9} 
$2$ & $r_{21}$ & $r_{22}$ &  &  &  &  &  & $p_{2}$\tabularnewline
\cline{2-9} 
$3$ &  &  & $r_{33}$ &  &  &  &  & \tabularnewline
\cline{2-9} 
$4$ &  &  &  & $r_{44}$ &  &  &  & \tabularnewline
\cline{2-9} 
$5$ &  &  &  &  & $r_{55}$ &  &  & \tabularnewline
\cline{2-9} 
$6$ &  &  &  &  &  & $r_{66}$ &  & \tabularnewline
\cline{2-9} 
$\vdots$ &  &  &  &  &  &  & $\vdots\vdots\vdots$ & $\vdots$\tabularnewline
\cline{2-9} 
 & $q_{1}$ & $q_{2}$ &  &  &  &  & $\ldots$ & \tabularnewline
\end{tabular}..

\end{center}

\noindent From (\ref{eq: 1 of 1-2})-(\ref{eq: 2 of 1-2}), $r_{11}+r_{12}+r_{21}+r_{22}=\min\left\{ p_{1}+p_{2}q{}_{1}+q_{2}\right\} $,
and since $r_{12}$ is $p$-minimized, $r_{11}+r_{12}+r_{21}+r_{22}=p_{1}+p_{2}$.
This means

\begin{center}

\begin{tabular}{c|c|c|c|c|c|c|c||c}
\multicolumn{1}{c}{} & \multicolumn{1}{c}{$1$} & \multicolumn{1}{c}{$2$} & \multicolumn{1}{c}{$3$} & \multicolumn{1}{c}{$4$} & \multicolumn{1}{c}{$5$} & \multicolumn{1}{c}{$6$} & \multicolumn{1}{c}{$\ldots$} & \tabularnewline
\cline{2-9} 
$1$ & $r_{11}$ & $r_{12}>0$ & $0$ & $0$ & $0$ & $0$ & \textbf{$\mathbf{0}$} & $p_{1}=r_{11}+r_{12}$\tabularnewline
\cline{2-9} 
$2$ & $r_{21}$ & $r_{22}$ & $0$ & $0$ & $0$ & $0$ & \textbf{$\mathbf{0}$} & $p_{2}=r_{21}+r_{22}$\tabularnewline
\cline{2-9} 
$3$ &  &  & $r_{33}$ &  &  &  &  & \tabularnewline
\cline{2-9} 
$4$ &  &  &  & $r_{44}$ &  &  &  & \tabularnewline
\cline{2-9} 
$5$ &  &  &  &  & $r_{55}$ &  &  & \tabularnewline
\cline{2-9} 
$6$ &  &  &  &  &  & $r_{66}$ &  & \tabularnewline
\cline{2-9} 
$\vdots$ &  &  &  &  &  &  & $\vdots\vdots\vdots$ & $\vdots$\tabularnewline
\cline{2-9} 
 & $q_{1}\geq r_{11}+r_{21}$ & $q_{2}\geq r_{12}+r_{22}$ &  &  &  &  & $\ldots$ & \tabularnewline
\end{tabular}.

\end{center}

\noindent We also should have

\begin{center}

\begin{tabular}{c|c|c|c|c|c|c|c||c}
\multicolumn{1}{c}{} & \multicolumn{1}{c}{$1$} & \multicolumn{1}{c}{$2$} & \multicolumn{1}{c}{$3$} & \multicolumn{1}{c}{$4$} & \multicolumn{1}{c}{$5$} & \multicolumn{1}{c}{$6$} & \multicolumn{1}{c}{$\ldots$} & \tabularnewline
\cline{2-9} 
$1$ & $r_{11}$ & $r_{12}>0$ & $0$ & $0$ & $0$ & $0$ & \textbf{$\mathbf{0}$} & $p_{1}=r_{11}+r_{12}$\tabularnewline
\cline{2-9} 
$2$ & $0$ & $r_{22}$ & $0$ & $0$ & $0$ & $0$ & \textbf{$\mathbf{0}$} & $p_{2}=r_{22}$\tabularnewline
\cline{2-9} 
$3$ & $0$ &  & $r_{33}$ &  &  &  &  & \tabularnewline
\cline{2-9} 
$4$ & $0$ &  &  & $r_{44}$ &  &  &  & \tabularnewline
\cline{2-9} 
$5$ & $0$ &  &  &  & $r_{55}$ &  &  & \tabularnewline
\cline{2-9} 
$6$ & $0$ &  &  &  &  & $r_{66}$ &  & \tabularnewline
\cline{2-9} 
$\vdots$ & \textbf{$\mathbf{0}$} &  &  &  &  &  & $\vdots\vdots\vdots$ & $\vdots$\tabularnewline
\cline{2-9} 
 & $q_{1}=r_{11}$ & $q_{2}\geq r_{12}+r_{22}$ &  &  &  &  & $\ldots$ & \tabularnewline
\cline{9-9} 
\end{tabular}

\end{center}

\noindent because $r_{11}=\min\left\{ p_{1},q_{1}\right\} $ and
$r_{11}<p_{1}$.

Generalizing, we have established the following rules:

(R1) If $r_{ij}>0$ and it is $p$-minimized, then all non-diagonal
elements in the rows $i$ and $j$ are zero except for $r_{ij}$,
and all non-diagonal elements in the column $i$ are zero.

(R2) (By symmetry, on exchanging $p$s and $q$s) If $r_{ij}>0$ and
it is $q$-minimized, then all non-diagonal elements in the columns
$i$ and $j$ are zero except for $r_{ij}$, and all non-diagonal
elements in the row $j$ are zero.

Returning to our special arrangement of the rows and columns, let
us prove now that all $r_{1j}$ with $j>2$ are $q$-minimized.  Assume
the contrary, and with no loss of generality, let $r_{15}=0$ be $p$-minimized.
This would mean that 
\[
r_{15}+r_{51}=p_{1}+p_{5}-r_{11}-r_{55}=r_{12}+p_{5}-r_{55}=0,
\]
which could only be true if $r_{12}=0$, which it is not. 

\begin{center}

\begin{tabular}{c|c|c|c|c|c|c|c||c}
\multicolumn{1}{c}{} & \multicolumn{1}{c}{$1$} & \multicolumn{1}{c}{$2$} & \multicolumn{1}{c}{$3$} & \multicolumn{1}{c}{$4$} & \multicolumn{1}{c}{$5$} & \multicolumn{1}{c}{$6$} & \multicolumn{1}{c}{$\ldots$} & \tabularnewline
\cline{2-9} 
$1$ & $r_{11}$ & $r_{12}>0$ & $\underset{q-min}{0}$ & $\underset{q-min}{0}$ & $\underset{q-min}{0}$ & $\underset{q-min}{0}$ & $\underset{q-min}{\mathbf{0}}$ & $p_{1}=r_{11}+r_{12}$\tabularnewline
\cline{2-9} 
$2$ & $0$ & $r_{22}$ & $0$ & $0$ & $0$ & $0$ & \textbf{$\mathbf{0}$} & $p_{2}=r_{22}$\tabularnewline
\cline{2-9} 
$3$ & $0$ &  & $r_{33}$ &  &  &  &  & \tabularnewline
\cline{2-9} 
$4$ & $0$ &  &  & $r_{44}$ &  &  &  & \tabularnewline
\cline{2-9} 
$5$ & $0$ &  &  &  & $r_{55}$ &  &  & $p_{5}$\tabularnewline
\cline{2-9} 
$6$ & $0$ &  &  &  &  & $r_{66}$ &  & \tabularnewline
\cline{2-9} 
$\vdots$ & \textbf{$\mathbf{0}$} &  &  &  &  &  & $\vdots\vdots\vdots$ & $\vdots$\tabularnewline
\cline{2-9} 
 & $q_{1}=r_{11}$ & $q_{2}\geq r_{12}+r_{22}$ &  &  &  &  & $\ldots$ & \tabularnewline
\cline{9-9} 
\end{tabular}

\end{center}

Generalizing, we have established two additional rules:

\noindent (R3) If $r_{ij}$ and $r_{ij'}$ are both $p$-minimized
(for pairwise distinct $i,j,j'$), then they are both zero (because
if one of them is not, say $r_{ij}>0$, then $r_{ij'}=0$ and it must
be $q$-minimized).

(R4) (By symmetry, on exchanging $p$s and $q$s) If $r_{ij}$ and
$r_{i'j}$ are both $q$-minimized (for pairwise distinct $i,i',j$),
then they are both zero.

Returning to our special arrangement of the rows and columns, it follows
that nowhere in the matrix can we have $r_{ij}>0$ ($i>2$) which
is $q$-minimized. Indeed, if $j>2$, then this would have contradicted
R4 (because the zeros in the first row are all $q$-minimized), and
if $j=2$, it would have contradicted R2 (because $r_{12}>0$).

Let us prove now that if $j>2$ and $i>2$ and $i\not=j$, then there
is no $r_{ij}>0$ that is $p$-minimized. Assume the contrary: $r_{ij}>0$
and $q$-minimized, and consider $r_{2i},r_{i2}$. With no loss of
generality, let $\left(i,j\right)$=$\left(4,6\right)$. In accordance
with R1, we fill in the 4th and the 6th rows with zeros, and we fill
in the 4th column with zeros too:

\begin{center}

\begin{tabular}{c|c|c|c|c|c|c|c||c}
\multicolumn{1}{c}{} & \multicolumn{1}{c}{$1$} & \multicolumn{1}{c}{$2$} & \multicolumn{1}{c}{$3$} & \multicolumn{1}{c}{$4$} & \multicolumn{1}{c}{$5$} & \multicolumn{1}{c}{$6$} & \multicolumn{1}{c}{$\ldots$} & \tabularnewline
\cline{2-9} 
$1$ & $r_{11}$ & $r_{12}>0$ & $0$ & $0$ & $0$ & $0$ & \textbf{$\mathbf{0}$} & $p_{1}=r_{11}+r_{12}$\tabularnewline
\cline{2-9} 
$2$ & $0$ & $r_{22}$ & $0$ & $0$ & $0$ & $0$ & \textbf{$\mathbf{0}$} & $p_{2}=r_{22}$\tabularnewline
\cline{2-9} 
$3$ & $0$ &  & $r_{33}$ & $0$ &  &  &  & \tabularnewline
\cline{2-9} 
$4$ & $0$ & $0$ & $0$ & $r_{44}$ & $0$ & $r_{46}>0$ & \textbf{$\mathbf{0}$} & $p_{4}=r_{44}+r_{46}$\tabularnewline
\cline{2-9} 
$5$ & $0$ &  &  & $0$ & $r_{55}$ &  &  & \tabularnewline
\cline{2-9} 
$6$ & $0$ & $0$ & $0$ & $r_{64}=0$ & $0$ & $r_{66}$ & \textbf{$\mathbf{0}$} & $p_{6}=r_{66}$\tabularnewline
\cline{2-9} 
$\vdots$ & \textbf{$\mathbf{0}$} &  &  & \textbf{$\mathbf{0}$} &  &  & $\vdots\vdots\vdots$ & $\vdots$\tabularnewline
\cline{2-9} 
 & $q_{1}=r_{11}$ & $q_{2}\geq r_{12}+r_{22}$ &  & $q_{4}=r_{44}$ &  & $q_{6}\geq r_{46}+r_{66}$ & $\ldots$ & \tabularnewline
\cline{9-9} 
\end{tabular}

\end{center}

\noindent Then $r_{24},r_{42}$ are both zero, whence $\min\left(p_{2}+p_{4},q_{2}+q_{4}\right)$
must equal $r_{22}+r_{44}$ to be a maximal coupling. But
\[
\min\left(p_{2}+p_{4},q_{2}+q_{4}\right)=\min\left(r_{22}+r_{44}+r_{46},r_{12}+r_{22}+r_{44}+x\right)>r_{22}+r_{44},
\]
since both $r_{12}$ and $r_{46}$ are positive, a contradiction. 

We come to the conclusion that the only positive non-diagonal elements
in the matrix can be in the column $2$ (and they are all $p$-minimized). 

\begin{center}

\begin{tabular}{c|c|c|c|c|c|c|c||c}
\multicolumn{1}{c}{} & \multicolumn{1}{c}{$1$} & \multicolumn{1}{c}{$2$} & \multicolumn{1}{c}{$3$} & \multicolumn{1}{c}{$4$} & \multicolumn{1}{c}{$5$} & \multicolumn{1}{c}{$6$} & \multicolumn{1}{c}{$\ldots$} & \tabularnewline
\cline{2-9} 
$1$ & $r_{11}$ & $r_{12}>0$ & $0$ & $0$ & $0$ & $0$ & \textbf{$\mathbf{0}$} & $p_{1}=r_{11}+r_{12}$\tabularnewline
\cline{2-9} 
$2$ & $0$ & $r_{22}$ & $0$ & $0$ & $0$ & $0$ & \textbf{$\mathbf{0}$} & $p_{2}=r_{22}$\tabularnewline
\cline{2-9} 
$3$ & $0$ & $r_{32}\geq0$ & $r_{33}$ & $0$ & $0$ & $0$ & \textbf{$\mathbf{0}$} & $p_{3}=r_{32}+r_{33}$\tabularnewline
\cline{2-9} 
$4$ & $0$ & $r_{42}\geq0$ & $0$ & $r_{44}$ & $0$ & $0$ & \textbf{$\mathbf{0}$} & $p_{4}=r_{42}+r_{44}$\tabularnewline
\cline{2-9} 
$5$ & $0$ & $r_{52}\geq0$ & $0$ & $0$ & $r_{55}$ & $0$ & \textbf{$\mathbf{0}$} & $p_{5}=r_{52}+r_{55}$\tabularnewline
\cline{2-9} 
$6$ & $0$ & $r_{62}\geq0$ & $0$ & $0$ & $0$ & $r_{66}$ & \textbf{$\mathbf{0}$} & $p_{6}=r_{62}+r_{66}$\tabularnewline
\cline{2-9} 
$\vdots$ & \textbf{$\mathbf{0}$} & $\vdots$ & \textbf{$\mathbf{0}$} & \textbf{$\mathbf{0}$} & \textbf{$\mathbf{0}$} & \textbf{$\mathbf{0}$} & $\vdots\vdots\vdots$ & $\vdots$\tabularnewline
\cline{2-9} 
 & $q_{1}=r_{11}$ & $q_{2}\geq r_{12}+r_{22}$ & $q_{3}=r_{33}$ & $q_{4}=r_{44}$ & $q_{5}=r_{55}$ & $q_{6}=r_{66}$ & $\ldots$ & \tabularnewline
\cline{9-9} 
\end{tabular}

\end{center}

Generalizing, let $r_{ij}>0$ and $i\not=j$. Then, if $r_{ij}$ is
$p$-minimized, all non-diagonal elements of the matrix outside column
$j$ are zero (and the non-diagonal elements in the $j$th column
are $p$-minimized); if $r_{ij}$ is $q$-minimized, then all non-diagonal
elements of the matrix outside row $i$ are zero (and the non-diagonal
elements in the $i$th row are $q$-minimized).

It is easy to check that such a construction is always internally
consistent.
\end{proof}
\rule[0.5ex]{0.9\columnwidth}{1pt}
\begin{cor*}[Section \ref{sec: A-systematic-study}, Corollary \ref{cor: 1-2-system}]
The 1-2 system for the original rvs $R_{1}^{1},R_{1}^{2}$ has a
maximally-connected coupling if and only if either $p_{i}>q_{i}$
for no more than one $i$ (this single possible $i$ being the single
fixed $i$ in the formulation of the theorem), or $p_{j}<q_{j}$ for
no more than one $j$ (this single possible $j$ being the single
fixed $j$ in the formulation of the theorem), $i,j\in\left\{ 1,\ldots,k\right\} $.
\end{cor*}
\begin{proof}[Proof of Corollary \ref{cor: 1-2-system}]
The ``only if'' part is obvious. To demonstrate the ``if'' part,
consider (without loss of generality) the arrangement

\begin{center}

\begin{tabular}{c|c|c|c|c|c|c|c||c}
\multicolumn{1}{c}{} & \multicolumn{1}{c}{$1$} & \multicolumn{1}{c}{$2$} & \multicolumn{1}{c}{$3$} & \multicolumn{1}{c}{$4$} & \multicolumn{1}{c}{$5$} & \multicolumn{1}{c}{$6$} & \multicolumn{1}{c}{$\ldots$} & \tabularnewline
\cline{2-9} 
$1$ &  &  &  &  &  &  & $\ldots$ & $p_{1}\geq q_{1}$\tabularnewline
\cline{2-9} 
$2$ &  &  &  &  &  &  & $\ldots$ & $p_{2}$\tabularnewline
\cline{2-9} 
$3$ &  &  &  &  &  &  & $\ldots$ & $p_{3}\geq q_{3}$\tabularnewline
\cline{2-9} 
$4$ &  &  &  &  &  &  & $\ldots$ & $p_{4}\geq q_{4}$\tabularnewline
\cline{2-9} 
$5$ &  &  &  &  &  &  & $\ldots$ & $p_{5}\geq q_{5}$\tabularnewline
\cline{2-9} 
$6$ &  &  &  &  &  &  & $\ldots$ & $p_{6}\geq q_{6}$\tabularnewline
\cline{2-9} 
$\vdots$ & $\vdots$ & $\vdots$ & $\vdots$ & $\vdots$ & $\vdots$ & $\vdots$ & $\vdots\vdots\vdots$ & $\vdots$\tabularnewline
\cline{2-9} 
 & $q_{1}$ & $q_{2}\geq p_{2}$ & $q_{3}$ & $q_{4}$ & $q_{5}$ & $q_{6}$ & $\ldots$ & \tabularnewline
\cline{9-9} 
\end{tabular}

\end{center}

\noindent and fill it in as

\begin{center}

\begin{tabular}{c|c|c|c|c|c|c|c||c}
\multicolumn{1}{c}{} & \multicolumn{1}{c}{$1$} & \multicolumn{1}{c}{$2$} & \multicolumn{1}{c}{$3$} & \multicolumn{1}{c}{$4$} & \multicolumn{1}{c}{$5$} & \multicolumn{1}{c}{$6$} & \multicolumn{1}{c}{$\ldots$} & \tabularnewline
\cline{2-9} 
$1$ & $q_{1}$ & $p_{1}-q_{1}$ & $0$ & $0$ & $0$ & $0$ & \textbf{$\mathbf{0}$} & $p_{1}\geq q_{1}$\tabularnewline
\cline{2-9} 
$2$ & $0$ & $p_{2}$ & $0$ & $0$ & $0$ & $0$ & \textbf{$\mathbf{0}$} & $p_{2}$\tabularnewline
\cline{2-9} 
$3$ & $0$ & $p_{3}-q_{3}$ & $q_{3}$ & $0$ & $0$ & $0$ & \textbf{$\mathbf{0}$} & $p_{3}\geq q_{3}$\tabularnewline
\cline{2-9} 
$4$ & $0$ & $p_{4}-q_{4}$ & $0$ & $q_{4}$ & $0$ & $0$ & \textbf{$\mathbf{0}$} & $p_{4}\geq q_{4}$\tabularnewline
\cline{2-9} 
$5$ & $0$ & $p_{5}-q_{5}$ & $0$ & $0$ & $q_{5}$ & $0$ & \textbf{$\mathbf{0}$} & $p_{5}\geq q_{5}$\tabularnewline
\cline{2-9} 
$6$ & $0$ & $p_{6}-q_{6}$ & $0$ & $0$ & $0$ & $q_{6}$ & \textbf{$\mathbf{0}$} & $p_{6}\geq q_{6}$\tabularnewline
\cline{2-9} 
$\vdots$ & \textbf{$\mathbf{0}$} & $\vdots$ & \textbf{$\mathbf{0}$} & \textbf{$\mathbf{0}$} & \textbf{$\mathbf{0}$} & \textbf{$\mathbf{0}$} & $\vdots\vdots\vdots$ & $\vdots$\tabularnewline
\cline{2-9} 
 & $q_{1}$ & $q_{2}\geq p_{2}$ & $q_{3}$ & $q_{4}$ & $q_{5}$ & $q_{6}$ & $\ldots$ & \tabularnewline
\cline{9-9} 
\end{tabular}

\end{center}

\noindent with the empty cells filled in with zeros. Check that (a)
all rows sum to the marginals; (b) the second column sums to 
\[
\sum_{i=1}^{k}p_{i}-\left(\sum_{i=1}^{k}q_{i}-q_{2}\right)=q_{2};
\]
(c) the rest of the columns sum to the marginals; (d) all $r_{ii}$
are $\min\left(p_{i},q_{i}\right)$; and (e) for all pairs $r_{ij}$
($i\not=j$) the sums $r_{ii}+r_{ij}+r_{ji}+r_{jj}$ equal $\min\left(p_{i}+p_{j},q_{i}+q_{j}\right)$.
The latter is proved by considering first all $j\not=2$, where it
is obvious, and then $j=2$ where the computation is, for $i\not=2$,
\[
r_{ii}+r_{i2}+r_{2i}+r_{22}=q_{i}+\left(p_{i}-q_{i}\right)+0+p_{2}=p_{i}+p_{2},
\]
as it should be because the values in the second column are to be
$p$-minimized. 
\end{proof}
\rule[0.5ex]{0.9\columnwidth}{1pt}
\begin{thm*}[Section \ref{sec: A-systematic-study}, Theorem \ref{thm: exists unique}]
The system $\D$ is noncontextual if and only if its 1-2 subsystem
is noncontextual, i.e., if and only if one of the $R_{1}^{1}$\emph{
and }$R_{1}^{2}$ nominally dominates the other.
\end{thm*}
\begin{proof}[Proof of Theorem \ref{thm: exists unique}]
The ``only if'' part is Theorem \ref{thm: 1-2 =000023 constraints}.
All we need to proof the ``if `` part is to check that the relation
(\ref{eq: relation}) holds. Assume the arrangement is as in the previous
corollary. Consider first any set $i_{1},\ldots,i_{m}$ that does
not include 2:
\[
\min\left(p_{i_{1}}+p_{i_{2}}+\ldots+p_{i_{m}},q_{i_{1}}+q_{i_{2}}+\ldots+q_{i_{m}}\right)=q_{i_{1}}+q_{i_{2}}+\ldots+q_{i_{m}},
\]
\[
\sum_{j=1}^{m}\min\left(p_{i_{j}},q_{i_{j}}\right)=q_{i_{1}}+q_{i_{2}}+\ldots+q_{i_{m}},
\]
\[
\min\left(p_{i_{j}}+p_{i_{j'}},q_{i_{j}}+q_{i_{j'}}\right)-\min\left(p_{i_{j}},q_{i_{j}}\right)-\min\left(p_{i_{j'}},q_{i_{j'}}\right)=0.
\]
So, (\ref{eq: relation}) holds. If one of the indices (let it be
$i_{1}$) is 2, then 
\[
q_{2}+q_{i_{2}}+\ldots+q_{i_{m}}=\left(p_{2}+\sum_{x\not=2}\left(p_{x}-q_{x}\right)\right)+q_{i_{2}}+\ldots+q_{i_{m}}>p_{2}+p_{i_{2}}+\ldots+p_{i_{m}},
\]
so
\[
\min\left(p_{2}+p_{i_{2}}+\ldots+p_{i_{m}},q_{2}+q_{i_{2}}+\ldots+q_{i_{m}}\right)=p_{2}+p_{i_{2}}+\ldots+p_{i_{m}}.
\]
We also have

\[
\sum_{j=1}^{m}\min\left(p_{i_{j}},q_{i_{j}}\right)=p_{2}+q_{i_{2}}+\ldots+q_{i_{m}},
\]
and for any $j\not=2,j'\not=2$,
\[
\min\left(p_{i_{j}}+p_{i_{j'}},q_{i_{j}}+q_{i_{j'}}\right)-\min\left(p_{i_{j}},q_{i_{j}}\right)-\min\left(p_{i_{j'}},q_{i_{j'}}\right)=0,
\]
\[
\min\left(p_{2}+p_{i_{j}},q_{2}+q_{i_{j}}\right)-\min\left(p_{2},q_{2}\right)-\min\left(p_{i_{j}},q_{i_{j}}\right)=p_{i_{j}}-q_{i_{j}}.
\]
Since index $i_{1}=2$ is paired with each of $i_{2},\ldots,i_{m}$
only once, the right-hand side in (\ref{eq: relation}) is
\[
p_{2}+q_{i_{2}}+\left(p_{i_{2}}-q_{i_{2}}\right)+\ldots+q_{i_{m}}+\left(p_{i_{m}}-q_{i_{m}}\right)=p_{2}+p_{i_{2}}+\ldots+p_{i_{m}}.
\]
\end{proof}

\end{document}